%% file: main.tex
\documentclass[11pt,letterpaper]{article}
\usepackage[utf8]{inputenc}
\usepackage{times}
\usepackage{url}
\usepackage{algorithm}
\usepackage[noend]{algpseudocode}
\usepackage{fullpage}
\usepackage{pslatex} % use times roman
\usepackage{mathdots} % provides \iddots
\usepackage{amsmath}
\usepackage{amssymb}
\usepackage{amsthm}
\usepackage{color}
\usepackage{array}
\usepackage{xy}
\usepackage{setspace}
\usepackage{varwidth}% http://ctan.org/pkg/varwidth
\usepackage{multicol}
\usepackage{hyperref}
\usepackage{cite}
\usepackage{accents}

\newcommand{\defn}[1]{\emph{{#1}}}

\usepackage{todonotes}

\newcommand{\poly}{\operatorname{poly}}
\newcommand{\polylog}{\operatorname{polylog}}
\newcommand{\E}{\mathbb{E}}

\newcommand{\var}{\operatorname{Var}}

\renewcommand{\paragraph}[1]{\vspace{.5 cm} \noindent \textbf{\mathbold #1} }

\makeatletter
\renewcommand\paragraph{\@startsection{paragraph}{4}{\z@}%
                                    {1.5ex \@plus1ex \@minus.2ex}%
                                    {-1em}%
                                    {\normalfont\normalsize\bfseries}}
\makeatother

\usepackage{thmtools}
\usepackage{thm-restate}

\newtheoremstyle{slanted}% <name>
{3pt}% <Space above>
{3pt}% <Space below>
{\slshape}% <Body font>
{}% <Indent amount>
{\bfseries}% <Theorem head font>
{.}% <Punctuation after theorem head>
{.5em}% <Space after theorem heading>
{}% <Theorem head spec (can be left empty, meaning `normal')>
\theoremstyle{slanted}
\newtheorem{theorem}{Theorem}%[section]
\newtheorem{lemma}[theorem]{Lemma}
\newtheorem{claim}[theorem]{Claim}
\newtheorem{proposition}[theorem]{Proposition}

\title{Balanced Allocations: The Heavily Loaded Case with Deletions}
\author{ Nikhil Bansal\footnote{University of Michigan, CS Dept., \url{bansal@gmail.com }. Supported in part by the NWO VICI grant 639.023.812.} \phantom{fff} William Kuszmaul\footnote{MIT CSAIL, \url{kuszmaul@mit.edu}.  Funded by a Fannie and John Hertz Fellowship and an NSF GRFP Fellowship. This research was also partially sponsored by the United States Air Force Research Laboratory and the United States Air Force Artificial Intelligence Accelerator and was accomplished under Cooperative Agreement Number FA8750-19-2-1000. The views and conclusions contained in this document are those of the authors and should not be interpreted as representing the official policies, either expressed or implied, of the United States Air Force or the U.S. Government. The U.S. Government is authorized to reproduce and distribute reprints for Government purposes notwithstanding any copyright notation herein.}}
\date{}

\begin{document}

\maketitle
\thispagestyle{empty}
\begin{abstract}
In the 2-choice allocation problem, $m$ balls are placed into $n$ bins, and each ball must choose between two random bins $i, j \in [n]$ that it has been assigned to. It has been known for more than two decades, that if each ball follows the \textsc{Greedy} strategy (i.e., always pick the less-full bin), then the maximum load will be $m/n + O(\log \log n)$ with high probability in $n$ (and $m / n + O(\log m)$ with high probability in $m$). It has remained an open question whether the same bounds hold in the \emph{dynamic} version of the same game, where balls are inserted/deleted with no more than $m$ balls present at a time. 

We show that, somewhat surprisingly, these bounds \emph{do not} hold in the dynamic setting: already on $4$ bins, there exists a sequence of insertions/deletions that cause the \textsc{Greedy} strategy to incur a maximum load of $m/4 + \Omega(\sqrt{m})$ with probability $\Omega(1)$---this is the same bound that one gets in the single-choice allocation model where each ball is assigned to a random bin!

This raises the question of whether \emph{any} 2-choice allocation strategy can offer a strong bound in the dynamic setting. Our second result answers this question in the affirmative: we present a new strategy, called \textsc{ModulatedGreedy}, that guarantees a maximum load of $m / n + O(\log m)$, at any given moment, with high probability in $m$. We also show how to generalize \textsc{ModulatedGreedy} to obtain dynamic guarantees for the $(1 + \beta)$-choice setting, and for the setting of balls-and-bins on a graph. 

Finally, we consider an extension of the dynamic setting in which balls can be \emph{reinserted} after they are deleted, and where the pair $i, j$ that a given ball uses is consistent across insertions. This seemingly small modification renders tight load balancing impossible: on 4 bins, any balls-and-bins strategy that is oblivious to the specific identities of balls being inserted/deleted \emph{must} allow for a maximum load of $m/4 + \poly(m)$ at some point in the first $\poly(m)$ insertions/deletions, with high probability in $m$. This is a remarkable departure from the $m = n$ case where the maximum load of $O(\log \log n)$ holds independently of whether reinsertions are allowed or not. 
\end{abstract}

\newpage
\setcounter{page}{1}

\allowdisplaybreaks
\input{intro}

\input{upperbound}
\input{greedylower}

\input{generallower}

\input{extensions}
\appendix

\input{appendix}

{\small{
\bibliographystyle{alpha}
\bibliography{references}
}}

\end{document}

%% file: intro.tex
\section{Introduction}

Randomized balls-into-bins processes \cite{MRS01,Wieder} serve as a useful abstraction for studying load-balancing problems, with applications such as scheduling, distributed systems, and data structures. The goal is to assign balls (e.g., tasks) to bins (e.g., machines) such that the balls are balanced as evenly as possible across the bins, where each individual ball may have only a few available random options for bins that it can be placed in. 

It is well known that, if $n$ balls are placed into $n$ bins using the classical \textsc{SingleChoice} rule, where each ball is placed independently in a uniformly random bin, then the maximum load is $\Theta(\log n/\log \log n)$ with probability $1 - 1 / \poly(n)$. 

\paragraph{The power of $2$-choices.}
In a seminal 1994 paper, Azar, Broder, Karlin and Upfal \cite{ABKU94} showed that under a seemingly minor modification, where for each ball
{\em two} bins are chosen independently and uniformly at random, and the ball is placed {\em greedily} in the least loaded of the two bins, 
the maximum load reduces to $\log \log n+O(1)$ with high probability in $n$.
In the decades since, this {\em power of 2-choices} paradigm has been extremely influential, with both theoretical (e.g., \cite{theoretical1, theoretical2, theoretical3, theoretical4, theoretical5}) and empirical (e.g., \cite{empirical1, empirical2, empirical3, empirical4, empirical5}) applications, and with a large literature on generalizations; see e.g.,~\cite{MRS01,Wieder}~for some excellent surveys.

\paragraph{The heavily-loaded case.}
Azar et al.'s result \cite{ABKU94} prompted researchers to consider the \defn{heavily-loaded case}, where $m \gg n$ balls are inserted into $n$ bins. The early techniques that were developed for the lightly-loaded setting (i.e., layered induction \cite{ABKU94}, witness trees \cite{Vocking99, CFMMSU98}, and differential-equation approaches \cite{mitzenmacher2001power, mitzenmacher1999studying}) struggled to deliver strong bounds in the heavily-loaded setting, and for several years the best known bound stood at $m/n + \log \log n + O(m/n)$ \cite{CFMMSU98, Vocking99}. 
If we define the \defn{overload} to be the amount by which the maximum load exceeds $m / n$, then this bound allows for an overload as large as $\log \log n + O(m / n)$---such a bound is useful if $m \approx n$, but  when $m \gg n \log n$, the bound becomes worse even than the standard bound offered by \textsc{SingleChoice} (i.e., an overload of $O(\sqrt{(m/n)\log n})$).

In a breakthrough result, Berenbrink, Czumaj, Steger and V{\"o}cking~\cite{BCSV00proc} showed how to use Markov-chain techniques to obtain a much stronger bound of $\log \log n + O(1)$  on the overload, with probability $1 - 1 / \poly(n)$. Thus, somewhat remarkably, the gap between the maximum and average loads in the heavily-loaded case \emph{is the same} as in the lightly-loaded case, with high probability in $n$.

When $m \gg n$, the $O(\log \log n)$ overload bound does not, in general, extend to hold with probability $1 - 1 / \poly(m)$ (i.e., w.h.p.~in the number of \emph{balls}). However, the known techniques can be used to achieve a quite strong (and, when $n = O(1)$, optimal) bound of $O(\log m)$ on the overload in this case.

%There have subsequently been Related results based on potential functions are also known \cite{ptw,tw14,los}.

\paragraph{The dynamic setting.}
In typical load-balancing and data-structures applications, however, the items can be both inserted and deleted dynamically over time. Here two natural models have been studied: (i) the \defn{insertion/deletion} model in which each insertion involves a new ball with independent random bin choices, and (ii) the \defn{reinsertion/deletion} model in which a ball can be {\em reinserted} after being deleted, and has the same two random bin choices each time it is reinserted. Although these two models may seem quite similar at first glance, we shall see later that the distinction is significant. 

Note that, whereas in the insertion-only setting, $m$ is set to be the total number of insertions, in the dynamic setting, $m$ is set to be an \emph{upper bound} on the number of balls that are present at any given moment (and the sequence of insertions/deletions may be infinite). The objective is to minimize the \emph{overload}, which is now defined as the amount by which the maximum load exceeds $m / n$ at any given moment.\footnote{It is tempting to define the overload to be the amount by which the maximum load exceeds $m(t) / n$, where $m(t)$ is the number of balls present at time $t$. However, the following (folklore) example demonstrates the flaw with such a definition: Suppose we insert $m$ balls (using an arbitrary insertion strategy), and then we delete a random $m/2$ of those balls. Since the $m/2$ deletions are random, even if the system was perfectly balanced after the initial $m$ insertions, the bin loads will typically be $m/2n \pm \sqrt{m/2n}$, and the maximum load will be $m(t) / n + \tilde{\Theta}(\sqrt{m / n})$, which is no better than the bound trivially achieved by \textsc{SingleChoice}.}

Azar et al.\cite{ABKU94}~considered the insertion/deletion model with $m = n$ and with \emph{random deletions}: that is, $n$ balls are inserted initially, and then there is an infinite sequence of alternating insertions/deletions, where each deletion removes a \emph{random} ball. They showed that, at any given moment, the \textsc{Greedy} strategy achieves a maximum load of $\log \log n + O(1)$, with high probability in $n$. 

Subsequent work has considered the more general setting where the insertions/deletions are determined by an \defn{oblivious-adversary} (i.e., an adversary that does not know the random choices of the algorithm), and where the only constraint on the adversary is that the number of balls in the system can never exceed $m$. Using the witness tree technique, first introduced by \cite{ColeMHMRSSV98}, Cole et al.~\cite{CFMMSU98} analyzed the reinsertion/deletion model with $m = n$, and established that the \textsc{Greedy} strategy guarantees a maximum load of $O(\log \log n)$ with high probability in $n$. Later, V\"{o}cking \cite{Vocking99} improved this to $\log \log n + O(1)$, which remarkably, matches the bound in the non-dynamic (insertion-only) case up to an additive $O(1)$ term. 

\paragraph{What about the dynamic heavily-loaded case?}
For more than two decades, it has remained an open question what the optimal bounds are in the \emph{heavily-loaded case} if we wish to support both insertions and deletions performed by an oblivious adversary. Besides obvious theoretical interest, the question also arises naturally in practice---for example, as a scheduling problem in which jobs arrive and depart over time, the number of jobs (balls) at any moment is much larger than the number $n$ of machines (bins), and the only guarantee on the arrivals/departures of jobs is an upperbound $m/n$ on the average load at any moment. 

%, e.g.,~as a scheduling problem in which jobs arrive and depart over time, and the number of jobs (balls) at any moment is much more than the number of machines (bins). 

The dynamic heavily-loaded setting was studied by Cole et al.~\cite{CFMMSU98} and V\"{o}cking \cite{Vocking99, Vocking03}, who showed that \textsc{Greedy} has overload $\log \log n + O(m/n)$ with high probability in $n$. 
But again this bound is already worse for $m \gg n\log n$ than the $O(\sqrt{(m/n)\log n})$ overload bound for \textsc{SingleChoice} (which also holds in the dynamic setting).

However, it is widely believed that \textsc{Greedy} should also achieve similar bounds in the dynamic heavily-loaded case as in the non-dynamic heavily-loaded case (i.e., an overload of $O(\log \log n)$ and $O(\log m)$, w.h.p.~in $n$ and $m$, respectively). The current limitation would seem to be a technical one: the witness-tree techniques that allow for us to analyze dynamic games with oblivious adversaries \cite{CFMMSU98, Vocking03} are incompatible with the techniques (i.e., Markov-chain \cite{BCSV00proc} and potential-function \cite{PTW10, los2022balanced, talwar2014balanced} arguments) that achieve strong bounds in the heavily-loaded case. 

In this work we prove new upper and lower bounds for the dynamic heavily-loaded case. We split our results into two parts, the first of which considers the insertion/deletion model, and the second of which considers the reinsertion/deletion model. 
 
\subsection{Results in the Insertion/Deletion Model}

We begin by considering the insertion/deletion model, that is, an oblivious adversary performs an arbitrary sequence of insertions/deletions subject only to the constraint that no more than $ m $ balls are present at a time.

% Recall that the objective here is to bound the maximum bin load over the baseline value $m/n$.\footnote{One may wonder, why not ask for  a guarantee with to the baseline value of $m(t)/n$ at time $t$, where $m(t)$ is the total number of balls currently in the system. However, the following (folklore) instance shows that any $2$-choice strategy can be as bad as \textsc{SingleChoice} for this objective: Insert $m$ balls, and then delete random $m/2$ balls. Clearly, as the $m/2$ deletions are random, even if the system was perfectly balanced after the initial $m$ insertions, the bin loads will be typically $m/2n \pm \sqrt{m/2n}$, exactly matching \textsc{SingleChoice}.} 

\paragraph{A lower bound for \textsc{Greedy}.}
We show that, somewhat surprisingly, the \textsc{Greedy} strategy actually \emph{does not} offer strong bounds in the dynamic heavily-loaded setting. 
In particular, already for $n = 4$ bins, there exists an oblivious sequence of insertions/deletions after which there is a maximum load of 
\[m / n + \Omega(\sqrt{m})\]
with probability $\Omega(1)$. In other words, the \textsc{Greedy} strategy is no better than \textsc{SingleChoice} in this setting!

Our result represents a remarkable departure from the lightly-loaded $m = n$ case, where \textsc{Greedy} achieves an optimal bound of $O(\log \log n)$ (even in the \emph{re}insertion/deletion model). The result also offers an explanation for why all previous attempts \cite{CFMMSU98,Vocking03} to analyze \textsc{Greedy} for large $m$ have yielded only relatively weak bounds. 

The high-level intuition behind our lower bound is as follows. Using \textsc{Greedy}, if some bin $i$ contains far fewer balls than the other bins, then there will be a contiguous time window during which all of the insertions are maximally biased towards bin $i$. But this means that, later on, the adversary can perform a sequence of deletions in which the balls being \emph{deleted} exhibit a strong bias towards being from bin $i$. In other words, the biases that \textsc{Greedy} exhibits during insertions can be thrown back at it by future deletions. 

We present the full construction in Section \ref{sec:greedy}. As a warmup, we first show a simpler (but already nontrivial) lower bound of $m/n +\Omega(m^{1/4})$ for $n=4$ bins in Section \ref{sec:greedysimpler}, and then give the full lower bound of $m / n + \Omega(m^{1/2})$ in Section \ref{sec:unevengreedy}. For ease of exposition we mostly focus on the case of $n=4$ --- however, we also show how to use our techniques to obtain a lower bound of $m / n + m^{1/4} / \poly(n)$ for general $n$. 

%$$m/n + \Omega(m^{1/2}/\poly(n))$ lower bound for general $n$.

\paragraph{The \textsc{ModulatedGreedy} algorithm.}
Of course, the above phenomenon is not isolated to the \textsc{Greedy} strategy. Any strategy that exhibits biases between bins is at risk of having those biases thrown back at it via future deletions. This raises a natural question: is it possible for \emph{any 2-choice allocation strategy} to beat the bounds trivially achieved in the single-choice model?

Our second result is a new algorithm called \textsc{ModulatedGreedy}, in the insertion/deletion model, that at any time,
with high probability in $m$, achieves a maximum load of
\[m / n + O(\log m).\]
This bound is optimal for any strategy that achieves high-probability bounds in $m$ (see Section \ref{sec:tightness}).

Given the choice between two bins $i$ and $j$, the \textsc{ModulatedGreedy} algorithm chooses between the bins probabilistically, based on how their loads compare. In particular, it carefully modulates its biases between bins so that the adversary is unable to find any non-trivial correlations between how balls are inserted.
Interestingly, the structure of \textsc{ModulatedGreedy} also allows for a direct combinatorial analysis, which proceeds by coupling the behavior of \textsc{ModulatedGreedy} to a seemingly different (and much simpler) randomized process that we call the {\em stone game}.

\paragraph{Generalizations.} Our analysis of \textsc{ModulatedGreedy} extends to support a number of generalizations and applications. This includes a tight bound of $m / n + O(\beta^{-1} \log m)$ for the \emph{$(1 + \beta)$-choice} version of the game \cite{peres2010}, where a $(1 - \beta)$-fraction of the balls are inserted using \textsc{SingleChoice} and only a $\beta$-fraction of the balls get two choices; a bound of $m/n + \polylog m$ for the dynamic balls-and-bins game on an undirected well-connected regular graphs \cite{bansal2021well, KP06}; and a bound of $m / n + O(\log M)$ for the setting in which $m$ is permitted to increase over time, subject only to the constraint that $m \le M$. In all of these settings, the previous states of the art were restricted to the insertion-only model.

To describe the main ideas as clearly as possible, we describe these results in two parts. In Section \ref{sec:upper} we consider a simpler version of \textsc{ModulatedGreedy} that guarantees the $m/n + O(\log m)$ bound for insertion/deletion sequences of $\poly(m)$ length. 
%This already contains the key coupling idea.
Later, in Section \ref{sec:applications}, we consider the general setting with unbounded request sequences and where $m$ can increase over time. The extensions to the $(1+\beta)$-choice and the graphical 2-choice processes are described in Section \ref{sec:ext}.

%In other words, the biases that \textsc{Greedy} exhibits during insertions can be thrown right back at it during deletions.

\subsection{An Impossibility Result for the Reinsertion/Deletion Model}

Finally, in Section \ref{sec:generallower}, we turn our attention to the reinsertion/deletion model. That is, the adversary can perform an arbitrary sequence of insertions, deletions, and reinsertions (as long as the ball being reinserted is not currently present) subject only to the constraint that no more than $m$ balls are present at a time. 

Here we establish an impossibility result. Consider any 2-choice bin-allocation strategy that is {\em oblivious} to the specific {\em identities} of balls (i.e., when a ball is inserted, all that the strategy gets to see is the pair $i, j$ of bins that the ball is assigned to). We show that, against any such strategy, it is possible for an oblivious adversary to force a maximum load of $m/4 + \poly(m)$ at some point in the first $\poly(m)$ insertions/deletions, with high probability in $m$.

% Say that such a strategy achieves overload $f(m)$ if, on any sequence of $\poly(m)$ insertions/deletions, the strategy guarantees a maximum load of at most $m / n + f(m)$, with high probability in $m$. We show that for any such strategy, there is an oblivious sequence of $\poly(m)$ reinsertions/deletions on $n=4$ bins that forces $f(m) \ge \poly(m)$.  That is, any {\em ID-oblivious} allocation rule must allow for an overload of $\poly(m)$ with probability at least $1 / \poly(m)$. 

This result reveals a fundamental (and perhaps unexpected) gap between the insertion/deletion model and the reinsertion/deletion model. 
In particular, in the lightly-loaded setting  with deletions where $m \leq n$, both models yield the same $O(\log \log n)$ bounds even for infinite sequences of reinsertions/deletions \cite{CFMMSU98,Vocking03}.
But, in the heavily-loaded setting, the cyclic dependencies that are introduced by reinsertions (i.e., a ball $x$ being reinserted is being placed into a system whose state has \emph{already} been affected by $x$'s bin choices in the past) end up being lethal to any ID-oblivious allocation strategy. 

\subsection{Other Related Work}

Beyond research on the heavily-loaded and dynamic settings, there has been a large body of work on other ways to extend the 2-choice allocation framework---because the literature on this subject is so extensive, we give only a brief overview here. These extensions have included work on restricted classes of insertion strategies (e.g., $(1 + \beta)$-choice strategies \cite{peres2010, weighted3}, thinning strategies \cite{los2022balanced,feldheim2021power,los2022balanced2}, strategies with limited information \cite{los2022balanced2}, etc.), on balls with nonuniform sizes \cite{talwar2014balanced, weighted2, weighted3, weighted4}, on parallel settings in which balls arrive in batches \cite{parallel1, parallel2, parallel3, parallel4, parallel5}, on settings in which bins correspond to vertices on a graph \cite{bansal2021well,KP06}, on settings where balls can be relocated after insertion \cite{aamand2021load, bender2021optimal}, etc. Another notable extension is V\"ocking's asymmetric $d$-choice paradigm \cite{Vocking03} which, in the lightly-loaded setting, chooses between $d$ bins on each insertion to achieve a maximum load of $O((\log \log n) / d)$.  

Another line of work, related to the current work on the dynamic setting, is on queuing models \cite{mitzenmacher2001power, vvedenskaya1996queueing, mukherjee2018universality, bramson2010randomized, luczak2006maximum, brightwell2012supermarket, eschenfeldt2016supermarket, luczak2005strong}, where insertions and deletions are {\em stochastic}. Many of these focus on the so-called supermarket model, introduced by \cite{mitzenmacher2001power, vvedenskaya1996queueing}, in which customers (i.e., balls) arrive in a Poisson stream of rate $\lambda n$, $\lambda < 1$, and are processed within each queue (i.e., bin) in FIFO order, where each customers requires processing time that is exponentially distributed with mean $1$. In the case where $\lambda$ is allowed to go to $1$ (see, e.g., \cite{brightwell2012supermarket, eschenfeldt2016supermarket}), the number of balls in the system can become $\omega(n)$ (this is analogous to the heavy case in standard balls and bins). However, because insertions/deletions are assumed to be stochastic, the analyses (and the flavors of the results) take a very different form than those in this paper (where deletions are performed by an oblivious adversary, and the number of balls in the system is deterministically bounded by a parameter $m$). 

In addition to the past work described above, there have also been recent efforts within the succinct-data-structure literature to obtain stronger bounds for the reinsertion/deletion model in specialized regimes, resulting in a 3-choice allocation scheme that achieves a bound of $m/n + O(\log \log n) + O(\sqrt{m / n} \cdot \sqrt{\log (m / n)})$ on the maximum load at any given moment \cite{bender2021tiny, bender2021all}. This bound is useful when $m \le O(n \log n)$, but does not improve significantly on \textsc{SingleChoice} when $m \gg n$.

\subsection{Preliminaries}
In the \defn{dynamic 2-choice allocation problem}, an oblivious adversary performs a sequence of ball insertions and deletions subject to the constraint that the number of balls in the system can never exceed $m$. 
 Whenever a ball $x$ is inserted, a uniformly random pair $h(x) = (h_1(x), h_2(x)) \in [n] \times [n]$ of distinct bins is selected, and the \defn{insertion strategy} must choose which of the bins $h_1(x)$ or $h_2(x)$ the ball will be placed in. The pair $h(x)$ is sometimes referred to as the \defn{hash} of the ball $x$.  

There are two models that we will consider for insertions and deletions. In the \defn{insertion/deletion model}, each insertion \textsc{Insert($x$)} places a new ball $x$ into the system that has never been present before. In the \defn{reinsertion/deletion model}, each insertion \textsc{Insert($x$)} places a ball $x$ into the system that is not  \emph{currently} present, but that may have been present in the past (each time $ x $ is inserted, its bin pair $h(x)$ stays the same). In both models, the \textsc{Delete($x$)} operation selects a ball $ x $ that is currently present and removes it.

We are interested in bounding the maximum \defn{load} (i.e., the number of balls) of any bin. Our algorithms will offer guarantees with high probability (w.h.p.) in $m$, meaning that the failure probability is $1 / \poly(m)$ for a polynomial of our choice. Two basic insertion strategies that we will discuss frequently are \textsc{Greedy}, which always selects the least full of the bins $h_1(x), h_2(x)$, and \textsc{SingleChoice}, which always selects bin $h_1(x)$.

In our lower bound for the reinsertion/deletion model (Section \ref{sec:generallower}), we will study the class of \defn{ID-oblivious} insertion strategies---such a strategy makes each insertion decision based on the hash $h(x)$ of the ball being inserted, rather than based on the specific identity $x$ of the ball. Formally, an ID-oblivious strategy is one that can be implemented with operations \textsc{Insert($h_1(x), h_2(x)$)} (indicating the pair of bins for the ball being inserted) and \textsc{Delete($r$)} (indicating a deletion of the $r$-th-most-recently-inserted ball of those present).

Finally, although  $h(x) = (h_1(x), h_2(x))$ is a uniformly random pair of \emph{distinct} bins, any strategy in the insertion/deletion model can choose to view $h(x)$ as a pair of \emph{independent bins} by artificially resetting $h_2(x) = h_1(x)$ with probability $1 / n$. The strategies that we design in this paper will assume (without loss of generality) that they are given a uniformly random pair of (not necessarily distinct) bins for each insertion.

%% file: upperbound.tex
\section{\textsc{ModulatedGreedy}: Handling $\poly(m)$ Insertions/Deletions}\label{sec:upper}

In this section, we consider the insertion/deletion model, with $n$ bins and up to $m$ balls present at a time, and we describe an insertion strategy, called \textsc{ModulatedGreedy}, that achieves a strong bound on maximum load. Here, we describe the simplest possible version of the strategy, which supports any sequence of $\poly(m)$ insertions/deletions while guaranteeing a maximum load of $m/n + O(\log m)$ with high probability in $m$. 
Later, in Section \ref{sec:applications}, we will extend \textsc{ModulatedGreedy} in various ways, such as supporting an infinite sequence of insertions/deletions, allowing $m$ to increase over time, etc.

%In this section, however, we will for the sake of brevity focus on the simplest possible version of \textsc{ModulatedGreedy}, which handles only a $\poly(m)$ number of insertions/deletions. 

% We assume here that \textsc{ModulatedGreedy} is given the parameter $m$\te at the beginning of the game, and that \textsc{ModulatedGreedy} is permitted to \emph{fail} with probability $1 / \poly(m)$ on each insertion---we will remove these assumptions and later subsections. 

The main result of the section is the following:

\begin{theorem}
Let $m \ge n$. Consider the insertion/deletion model with $n$ bins and an upper bound of at most $m$ balls present at a time. Consider a sequence of $\poly(m)$ insertions/deletions, where insertions are implemented using \textsc{ModulatedGreedy}. With high probability in $m$, \textsc{ModulatedGreedy} does not halt during any of the insertions/deletions, and no bin ever has load more than $m/n + O(\log m)$.
\label{thm:modgreedy}
\end{theorem}

When we describe the lower bound for \textsc{Greedy} in Section \ref{sec:greedy}, we will see  that the main problem with \textsc{Greedy} is that it is too aggressive. Given the choice between two bins $i, j$, as \textsc{Greedy} always chooses the less loaded of the two---this creates correlations between balls that can be exploited to construct a bad sequence of insertions/deletions. In contrast, \textsc{ModulatedGreedy} will try to be as \emph{unaggressive} as possible, while still guaranteeing an upper gap of $O(\log m)$. In particular, it carefully modulates its behavior and only exhibits a strong bias between two bins $i$ and $j$ if (1) the two bins $i$ and $j$ have significantly different loads; and (2) the system is nearly saturated (i.e., there are nearly $m$ balls present). 

As we shall see, this modulated behavior also allows for a simple (but clever) combinatorial analysis, marking a departure from the (typically quite involved) potential-function and Markov-chain arguments used in past analyses of the heavily-loaded case.

\subsection{The Algorithm} 
The \textsc{ModulatedGreedy} algorithm for allocating a bin to a ball is given below. We assume without loss of generality that $m$ is a multiple of $n$.
%as Algorithm \ref{alg:modulated1}.
\begin{algorithm}
\begin{algorithmic}
\Procedure{ModulatedGreedy}{} 
\State Select two bins $i, j \in [n]$ independently and uniformly at random. 
\State Set $T = m/n + c \log m - \sum_r \ell_r / n$. 
\If{$ (\max_k \ell_k) - (\min_k \ell_k) \le T$}
  \State{Assign the ball to bin $i$ with probability $1/2 + \frac{\ell_j - \ell_i}{2T}$, and otherwise assign it to bin $j$.}
\Else
  \State{Halt.}
\EndIf
\EndProcedure
\end{algorithmic}
\label{alg:modulated1}
\caption{The \textsc{ModulatedGreedy} insertion strategy. Here, $\ell_k$ is the number of balls in bin $k$ prior to the insertion, and $c$ is a large positive constant.} 
\end{algorithm}

For $ k\in [n]$, let $\ell_k$ denote the load on bin $k$ prior to the insertion, let
$\overline{\ell} = \sum_k \ell_k/n$ be the average bin load, and $c$ be a (sufficiently large) fixed constant. 
When choosing between two bins $i, j$, the algorithm exhibits bias
\[(\ell_j - \ell_i)/2T\]
towards bin $i$, where  \[T = m/n + c \log m - \overline{\ell}.\] 
Note that the algorithm is well-defined as long as $|\ell_j - \ell_i| \leq T$ for all $i,j \in [n]$.
One should think of $T$ as representing the average amount of leftover space that each bin would have if each bin had a total capacity of $m/n + c \log m$ balls. This means that the bias is proportional to the difference $\ell_j - \ell_i$ between the loads of the bins, and is inversely proportional to the average amount $T$ of space left in each bin. 

% {\bf Remark.} The algorithm can also be implemented as a thinning algorithm\footnote{In {\em thinning} \cite{los2022balanced,feldheim2021power,los2022balanced2}: to insert a ball, we select a random bin $i$ and if its load is below some threshold (that we choose);  we place it in bin $i$, and otherwise we place it in a random bin $j$. If we select the threshold $t$ uniformly in $ [(\min_k \ell_k) , (\min_k \ell_k) + T]$ (the probability that we pick bin $i$ is precisely  \ell_i-\ell_\min /T  + (1-p)/n 

% $1/ 2 + \frac{\ell_j - \ell_i}{2T}$);}.
%\nikhil{I think we might have to change algorithm to modify the description to halt whenever some $\ell_j \leq \overline{\ell} - T$. Skipping this for now. Otherwise, I don't see how to specify a fixed range from which to pick the random threshold.}

The following lemma gives a closed-form solution for the probability of a given bin $k$ being selected by \textsc{ModulatedGreedy}.
\begin{lemma}
Suppose that $|\ell_i - \ell_j| \le T$ for all bins $i, j$. Consider a bin $k$, and set $T_k = m/n + c \log m - \ell_k$. Upon an insertion, a bin $k$  is selected with probability $T_k/(nT) = T_k/(\sum_i T_i)$.
\label{lem:modulateddist}
\end{lemma}
% \todo{May want to move this proof to an appendix}
% \nikhil{yes}
\begin{proof}
Let $i, j$ denote the random bin choices for the ball being inserted. The probability that a given bin $k$ is selected is given by
\begin{align*} 
& \Pr[i, j = k] + \sum_{s \neq k} \Pr[i = k, j = s] \left(\frac{1}{2} + \frac{\ell_s - \ell_k}{2T}\right) + \sum_{s \neq k} \Pr[i = s, j = k] \left(\frac{1}{2} + \frac{\ell_s - \ell_k}{2T}\right) \\
& = \frac{1}{n^2} + \frac{2}{n^2} \sum_{s \neq k} \left(\frac{1}{2} + \frac{\ell_s - \ell_k}{2T}\right)  = \frac{2}{n^2} \sum_{s = 1}^n \left(\frac{1}{2} + \frac{\ell_s - \ell_k}{2T}\right) \\
& = \frac{2}{n} \left(\frac{1}{2} + \frac{\overline{\ell} - \ell_k}{2T}\right) 
 = \frac{T + \overline{\ell} - \ell_k}{n T} 
%& = \frac{(c\log m + m/n - \overline{\ell}) + \overline{\ell} - \ell_k}{n \cdot T}  \\
%& = \frac{c \log m + m/n - \ell_k}{n \cdot T}\\
 = \frac{T_k}{n  T}.
\end{align*}
Finally we note that $\sum_{i=1}^n T_i = \sum_{i=1}^n  (m/n + c \log m - \ell_i) = m + n c \log m - n\overline{\ell}   = nT$.
\end{proof}
% \begin{proof}
% Define $\overline{\ell} = (\sum_s \ell_s) / n$ to be the average of the $\ell_s$'s. Let $i, j$ be the random bins selected by the ball being inserted. The probability that bin $k$ is selected is given by
% \begin{align*} 
% & \Pr[i, j = k] + \sum_{s \neq k} \Pr[i = k, j = s] \left(\frac{1}{2} + \frac{\ell_s - \ell_k}{2T}\right) + \sum_{s \neq k} \Pr[i = s, j = k] \left(\frac{1}{2} + \frac{\ell_s - \ell_k}{2T}\right) \\
% & = \frac{1}{n^2} + \frac{2}{n^2} \sum_{s \neq k} \left(\frac{1}{2} + \frac{\ell_s - \ell_i}{2T}\right) \\
% & = \frac{2}{n^2} \sum_{s = 1}^n \left(\frac{1}{2} + \frac{\ell_s - \ell_k}{2T}\right) \\
% & = \frac{2}{n} \left(\frac{1}{2} + \frac{\overline{\ell} - \ell_k}{2T}\right) \\
% & = \frac{T + \overline{\ell} - \ell_k}{n \cdot T}  \\
% & = \frac{(c\log m + m/n - \overline{\ell}) + \overline{\ell} - \ell_k}{n \cdot T}  \\
% & = \frac{c \log m + m/n - \ell_k}{n \cdot T}\\
% & = \frac{T_k}{n \cdot T}.
% \end{align*}
% \end{proof}

\subsection{Analysis}
To analyze \textsc{ModulatedGreedy}, we begin by describing a seemingly different process (which we call the stone game) that, by design, yields to a simple combinatorial analysis. We then show that the \textsc{ModulatedGreedy} algorithm and the stone game can be \emph{coupled together} so that bounds on the behavior of the stone game directly imply bounds on the behavior of \textsc{ModulatedGreedy}.

\paragraph{Stone Game.}
In the $(Q,n)$-{\em stone game}, parameterized by $Q$ and $n$, there are $Qn$ stones which are distributed among two bags; an {\em inactive bag} and an {\em active bag}.
Initially the active bag is empty, and all the stones are in the inactive bag. 

The game supports two types of operations: the \textsc{Activate()} operation moves a \emph{random} stone from the inactive bag to the active bag; and the \textsc{Deactivate($r$)} operation examines the stones in the active bag, selects the stone that was added the $r$-th most recently, and moves it back to the inactive bag. (\textsc{Activate()} can only be called if the inactive bag is non-empty, and \textsc{Deactivate($r$)} can only be called if the active bag contains $r$ or more balls). The sequence of operations is generated by an oblivious adversary, independent of the random bits used by the game.

The stones are labeled $x_{k,q}$ for $k\in [n],q \in [Q]$. We call $k$ the \emph{color} of the stone, so that there are $Q$ stones of each color. However, the labels of the stone should be thought of as {\em hidden}, since the behaviors of \textsc{Activate()} and \textsc{Deactivate($r$)} do not depend on the labels of the stones.
%there are are $Kn$ stones $\{x_{i, j}\}_{i \in [n], j \in [K]}$ 

We will now prove some lemmas establishing that the stone game is, by design, very well behaved.  
Our first lemma shows that, even though the adversary gets to perform activations/deactivations, it has no control over which specific stones are in the active bag.
\begin{lemma}
At any given moment, if the active/inactive bag contains $s$ stones, then these stones are a uniformly random subset of size $s$ of the stones $\{x_{k, q}\}_{k \in [n], q \in [Q]}$.
\label{lem:randomactive}
\end{lemma}
\begin{proof}
The point is that the activation/deactivation operations do not depend on the labels of the balls. 

Formally, fix any sequence of activations/deactivations and the random choices of the \textsc{Activate()} operations, and let $S$ be set of stones currently in the inactive bag (the argument for the active bag is identical).
%and let $I$ be the transcript of which ball is moved during each activation/deactivation. Note that $I$ is a random variable, so $I \sim \mathbf{I}$ for some probability distribution $\mathbf{I}$ over transcripts.
Then for any run of the game with a random permutation $\pi$ applied to the $Qn$ labels $\{x_{k, q}\}_{k \in [n], q \in [Q]}$,
%For any transcript $I \sim \mathbf{I}$ of the stone game, we can construct a second instance $\pi(I)$ that is permuted by $\pi$: whenever \textsc{Activate()} or \textsc{Deactivate($r$)} selects a stone $x_{i, j}$ in $I$, the same operation selects $\pi(x_{i, j})$ in $\pi(I)$. 
%The map $\pi:I \rightarrow \pi(I)$ is a bijection from $\mathbf{I}$ to $\mathbf{I}$. Thus, if we wish to analyze the probability distribution of $I \sim \mathbf{I}$, it suffices to analyze $\pi(I) \sim \mathbf{I}$. On the other hand, 
the set stones in the active bag will be $\pi(S)$. Thus, if the inactive bag contains $s$ stones, every $s$-element subset of the $nQ$ stones is equally likely.
% Fix a sequence of activation/deactivations, and let $I$ be the transcript of which ball is moved during each activation/deactivation. Note that $I$ is a random variable, so $I \sim \mathbf{I}$ for some probability distribution $\mathbf{I}$ over transcripts.
%
% Let $\pi$ be a random permutation on $\{x_{i, j}\}_{i \in [n], j \in [K]}$. For any transcript $I \sim \mathbf{I}$ of the stone game, we can construct a second instance $\pi(I)$ that is permuted by $\pi$: whenever \textsc{Activate()} or \textsc{Deactivate($r$)} selects a stone $x_{i, j}$ in $I$, the same operation selects $\pi(x_{i, j})$ in $\pi(I)$. 
%
% The map $\pi:I \rightarrow \pi(I)$ is a bijection from $\mathbf{I}$ to $\mathbf{I}$. Thus, if we wish to analyze the probability distribution of $I \sim \mathbf{I}$, it suffices to analyze $\pi(I) \sim \mathbf{I}$. On the other hand, $\pi(I)$ randomly permuted which stones are in the active bag at the end of the game. Thus, if the active bag contains $q$ stones the end of the game, every $q$-element subset of $\{x_{i, j}\}_{i \in [n], j \in [K]}$ is equally likely to be the set of stones in the bag.
\end{proof}

%Define the \defn{color} of a stone $x_{i, j}$, $i \in [n], j \in [K]$ to be the number $i$. 
This implies that as long as the inactive bag contains a reasonably large number of stones (namely, $\Omega(n \log (nQ))$), each color is guaranteed to have roughly equal representation in the bag.

\begin{lemma}
Suppose at some given moment, the inactive bag contains $s \ge c n \log (nQ)$ stones, for some large enough constant $c$. Let $s_k$ be the number of these stones with color $k$. Then  
$s_k \in [s/2n, 3s/2n]$ for each $k \in [n]$, with probability at least $1-1/(Qn)^{\Omega(c)}$.  
\label{lem:stonesbycolor}
\end{lemma}
\begin{proof}
By Lemma \ref{lem:randomactive}, the balls $S$ in the inactive bag are a random subset of size $s$ of the $Qn$ balls $\{x_{k, q}\}$. Let $X_k = \{x_{k, 1}, \ldots, x_{k, Q}\}$ be the set of all color-$k$ balls. Then $s_k = |X_k \cap S|$, the number $s_k$ of balls of color $k$ in $S$, has the hypergeometric distribution $H(Qn,Q,s)$.

As the standard tail bounds on sampling without replacement at least as sharp as those given by Chernoff bounds for sampling with replacement \cite{FK15} (Section 22.5), and as $\E[s_i]=s/n$, we get that 
%Suppose we select $y_1, y_2, \ldots, y_{\alpha}$ one after another (so $y_i$ is random from $\{x_{i, j}\} \setminus \{y_1, \ldots, y_{i - 1}\}$), and we define $Y_i$ to be the indicator random variable for the event that $y_i$ has color $q$. Then $\alpha_q = \sum_i Y_i$ has expected value $\alpha / n$ and is a sum of negatively correlated indicator random variables---by a Chernoff bound (see e.g., \cite{impagliazzo2010constructive}) we have that 
\begin{equation}
\Pr[|s_k - s/n| \geq \epsilon s/n]   \le  2 \exp(-\epsilon^2 s/3n).
\label{eq:alphaconc}
\end{equation}
Setting $\epsilon = 1/2$, and taking a union bound over the $n$ colors, gives  that $s_k \in [s/2n, 3s/2n]$  for each $k \in [n]$
with probability $1 - 2 n \exp(-\Omega(c \log Qn))$ which is $1- 1/(Qn)^{\Omega(c)}$ for large enough $c$.
%Since $\alpha \ge c n \log (nK)$ for a sufficiently large positive constant $ c $, it follows that \eqref{eq:alphaconc} holds with high pro bability in $ nK $, as desired.
\end{proof}

\subsubsection{Relating the stone game to  the balls-and-bins game} One can think of the stones in the stone game as being similar to balls in the balls-and-bins game---the active bag represents the set of balls that are present, the color of a stone dictates which ``bin'' a given ball is in, and activations/deactivations correspond to insertions/deletions.

However, there are several significant differences between the games. Notably, the whole point of the balls-and-bins game is to ensure that no single bin contains too many balls, but in the stone game, the active bag trivially (and deterministically) has at most $Q$ stones of any given color. Nonetheless, we shall now see how to couple the two games together in such a way that our analysis of the stone game yields a bound for the balls-and-bins game.

\paragraph{Mapping between instances.}
We first giving a mapping between the sequence of insertions/deletions for balls-and-bins game and the input sequence for the stone game. For any sequence $\mathcal{S}$ of insertions/deletions in balls-and-bins game, define $\phi(\mathcal{S})$ to be a corresponding sequence of activations/deactivations, where each \textsc{Insert} operation is replaced with an \textsc{Allocate} operation, and where each \textsc{Delete($x$)} operation on a ball $x$ is replaced with a \textsc{Deactivate($r$)} operation, where $r - 1$ is the number of balls in the system that were inserted after $x$.

The following key lemma shows that the random choices in the two games can be coupled.
\begin{lemma} [Coupling]
\label{lem:coupling}
Let $n \le m$ and let $\Delta = c \log m$, where $c$ is the positive constant used by \textsc{ModulatedGreedy}. Consider a sequence $\mathcal{S}$ of insertions/deletions in a balls-and-bins game on $n$ bins, where there are never more than $m$ balls present at a time. Let $G_1$ be a balls-and-bins game with operation-sequence $\mathcal{S}$ and let $G_2$ be $(Q,n)$-stone game with $Q = m/n+\Delta$ with operation sequence $\phi(\mathcal{S})$. 

If $G_1$ is implemented using \textsc{ModulatedGreedy}, then there exists a coupling between $G_1$ and $G_2$ with the following property: Up until \textsc{ModulatedGreedy} halts, the number of balls in a given bin $k$ (in the balls-and-bins game) always equals the number of stones in the active bag with color $k$ (in the stone game).
% Consider a sequence of $\poly(m)$ insertions/deletions in a balls-and-bins game on $n$ bins, with up to $m$ balls present at a time, and suppose that insertions are implemented using the \textsc{ModulatedGreedy} algorithm. 
% Also consider a stone game on $nK = m + nc \log m$ balls, where $c$ is the constant used in the \textsc{ModulatedGreedy} algorithm. Suppose that, as we perform insertion/deletions in the balls-and-bins game, we also perform activations/deactivations in the stone game as follows: each insertion in the balls-and-bins game causes an activation in the stone game, and each deletion \textsc{Delete($x$)} in the balls-and-bins game causes a deactivation \textsc{Deactivate($r$)}, where $r - 1$ is the number of balls present in the balls-and-bins game that were inserted after $x$. 
% \nikhil{It might be better to define the instance of stone game without reference to modulated greedy, to make explicit the obliviousness: i.e. it is a deterministic mapping of the insertion/deletion instances to a stone game instance. Then just make this lemma about the coupling part(the lines below).} 
\end{lemma}
\begin{proof}
Let $\ell_1, \ell_2, \ldots, \ell_n$ denote the loads of the bins at any given moment.
By Lemma \ref{lem:modulateddist}, we know that, on any given insertion in which \textsc{ModulatedGreedy} does not halt, each bin $k$ is selected with probability
\begin{equation}
\frac{T_k}{nT} = \frac{T_k}{\sum_{i=1}^n T_i}.
%\frac{\frac{m}{n} + c \log m - \ell_k}{m + c n\log m - \sum_i \ell_i}
\label{eq:modpr}
\end{equation}
%where we use that  $\sum_i T_i = \sum_i (m/n + \Delta - \ell_i) = m - nc \log m - n \overline{\ell} = nT$.

%\nikhil{Could define the notation $d_k = m/n + c \log m -\ell_k$ (deficit in bin $k$). This would simplify computations and algo description. I think we already called it $T_k$.}
Now suppose that, for each color $k$ there are $\ell_k$ stones with color $k$ in the active bag (and hence $Q-\ell_k$ such stones in the inactive bag) of the stone game. Then on any given activation, the probability of a ball with color $k$ being moved into the active bag is 
\begin{equation}
\frac{Q - \ell_k}{nQ - \sum_{i=1}^n \ell_i} = \frac{m/n + \Delta - \ell_k}{m +  n \Delta  - \sum_i \ell_i} = \frac{T_k}{\sum_{i=1}^n T_i},
\label{eq:stonepr}
\end{equation}
where the first equality uses that $Q = m/n + \Delta$.
The two probabilities \eqref{eq:modpr} and \eqref{eq:stonepr} are precisely equal. Thus, we can couple the games so that the bin selected by the insertion in the balls-and-bins game is the same as the stone color selected by the activation in the stone game. 

If we implement the insertions/activations in this way, then the deletions/deactivations also become coupled: whenever a ball is deleted from a bin $k$, a stone with color $k$ is removed from the active bag (in particular, the ball and stone were assigned to have the same bin/color when they were inserted/activated previously). Thus the proof of the lemma is complete.
\end{proof}

\paragraph{Proof of Theorem \ref{thm:modgreedy}.}
Finally, we can use the coupling in Lemma \ref{lem:coupling} to bound the probability of \textsc{ModulatedGreedy} halting and prove Theorem \ref{thm:modgreedy}.
% \begin{theorem}
% Let $m \ge n$. Consider a sequence of $\poly(m)$ insertions/deletions in a balls-and-bins game on $n$ bins, with up to $m$ balls present at a time, and suppose that insertions are implemented using the \textsc{ModulatedGreedy} algorithm. With high probability in $m$, \textsc{ModulatedGreedy} does not halt during any of the insertions/deletions, and there are never more than $m/n + O(\log m)$ balls in any given bin.
% \label{thm:modgreedy}
% \end{theorem}
\begin{proof} (Theorem \ref{thm:modgreedy})
Observe that, if \textsc{ModulatedGreedy} does not halt, then deterministically there are at most  $m/n + O(\log m)$ balls in any given bin.
In particular, the condition $\max_k \ell_k - \min_k \ell_k \leq T$ implies that 
$\max_k \ell_k - \overline{\ell} \leq T$. Plugging $T= m/n + c \log m - \overline{\ell}$, this gives that
$\max_k \ell_k \leq m/n + c \log m$.

Thus, it suffices to analyze the probability of halting. 
%\nikhil{issue here.  should change the condition of halting... Currently, in the description it can happen  that we always have $i=1,j=2$ and $\ell_1 = \ell_2$  and $\ell_1,\ell_2> m/n +c \log m$, but the keeps running if $T > 0$ due to other underloaded bins.}

By Lemma \ref{lem:coupling}, up until \textsc{ModulatedGreedy} halts, it can be coupled to a stone game on $nQ = m + nc \log m$ balls, where the number of balls in the active bag never exceeds $m$. 
Under this coupling, the number of balls $\ell_k$ in bin $k$ satisfies $\ell_k = Q-s_k$, where $s_k$ is the number of color-$k$ stones in the inactive bag. 

The \textsc{ModulatedGreedy} algorithm halts only if 
\begin{equation}
|\ell_i - \ell_j| > T = m/n + c \log m - \overline{\ell}  = Q-\overline{\ell} 
\label{eq:haltcond}
\end{equation}
for some pair $i, j$ of bins. 
For the stone game, denoting $s= \sum_k s_k = \sum_k (Q-\ell_k) = n (Q-\overline{\ell})$,
and as $|s_i-s_j|= |\ell_i - \ell_j|$,  condition \eqref{eq:haltcond} is equivalent to \[|s_i - s_j| > s/n.\]
But we know by Lemma \ref{lem:stonesbycolor} that, w.h.p. in $m$, we have $|s_i - s_j| \le  s/n$
at all times during the stone game (since the number of balls in the inactive bag is always at least $n c \log m$). Thus, we have w.h.p. in $m$ that \textsc{ModulatedGreedy} never halts.\end{proof}

\subsection{Tightness of the Bound}
\label{sec:tightness}
Clearly, the bound of $m / n + O(\log m)$ is not optimal for all parameter regimes, since it is known that \textsc{Greedy} achieves maximum load $O(\log \log n)$ in the regime of $n = m$. We remark, however, that for parameter regimes where $m$ is much larger than $n$, or when $n$ is fixed, this bound is essentially optimal.
\begin{proposition}
Consider $m$ insertions into $4$ bins using any sequential 2-choice insertion strategy. With probability at least $1 / \poly(m)$, some bin contains at least $m / 4 + \Omega(\log m)$ balls. More generally for $n$ bins, some bin contains at least $m/n + \Omega(\log_n m)$.
\end{proposition} 
\begin{proof}
Let us consider the final $\log m$ insertions $x_1, \ldots, x_{\log m}$. Suppose, without loss of generality, that prior to those insertions being performed, bins $1, 2$ contain at least as many total balls as bins $3, 4$. With probability $1/\poly(m)$, all of the insertions $x_1, \ldots, x_{\log m}$ are forced to choose between bins $1$ and $2$. No matter how they are assigned, this forces at least one of bins  $1, 2$ to have load $m / 4 + \Omega(\log m)$ balls at the end of the insertions. 

Deleting all the balls, and repeating the instance again $\text{poly}(m)$ times, this event will occur for one of the instances with high probability.

The same argument also implies a $m/n + \Omega(\log_n m)$ lower bound; consider the final $k=\log_n m$ balls, and note that with probability $ \Omega(n^{-2k}) = 1/\text{poly}(m)$ the only bin choices for these balls are $1$ and $2$.
\end{proof}

%% file: greedylower.tex
\section{A Lower Bound for Greedy with Deletions}\label{sec:greedy}

This section gives a lower bound for the \textsc{Greedy} algorithm in the insertion/deletion model against an oblivious adversary, with up to $m$ balls present at a time. 
Recall that the trivial \textsc{SingleChoice} strategy achieves an overload of $O(\sqrt{(m/n) \log n})$ (w.h.p. in $m$) in this setting, so the natural question is whether \textsc{Greedy} does any better. We show that, even for $n = 4$ (meaning that \textsc{SingleChoice} has an overload of $O(\sqrt{n})$), it does not.  
%We show that, with probability $\Omega(1)$, the adversary can force some bin to contain $m / 4 + \Omega(\sqrt{m})$ balls (within $\poly(m)$ steps).

\begin{theorem}
\label{thm:greedy}
Consider the insertion/deletion model on $n = 4$ bins, with the restriction that at most $m$ balls can be present at any time, and suppose that insertions are implemented using \textsc{Greedy}. There exists an oblivious sequence of $\poly(m)$ insertions/deletions such that, after the sequence is complete, we have with probability $\Omega(1)$ that some bin contains $m/4 + \Omega(\sqrt{m})$ balls.
\end{theorem}

For ease of exposition, and to keep the main ideas as clear as possible, we focus our lower bound on $n=4$ bins. We will also see, however, that for general $n$ and $m$, a similar construction gives an $m/n + \Omega(\sqrt{m}/\text{poly}(n))$ lower bound on the maximum load.

\subsection{ A Simpler $\Omega(m^{1/4})$ Bound}
\label{sec:greedysimpler}
Before proving Theorem \ref{thm:greedy}, we first describe a simpler (but already surprisingly) lower bound of $m/4 + \Omega(m^{1/4})$. Later in Section \ref{sec:unevengreedy} we build on these ideas to prove Theorem \ref{thm:greedy}.

%\subsubsection{Weaker $\Omega(m^{1/4})$ lower bound on overload} 
We first describe a construction with the property that if we ever reach a state where one of the bins (say, bin 1) contains significantly fewer balls (say, $k$ fewer balls) than the other bins, then we can subsequently reach a state in which (with probability $\Omega(1)$), some bin contains at least $m/4 + \Omega(\sqrt{k})$ balls.  As we shall see later in the subsection, this can be used to directly obtain the $m/4 + \Omega(m^{1/4})$ bound.

\begin{proposition}[Gap to overload]
\label{prop:greedygap}
Consider the \textsc{Greedy} algorithm on 4 bins, on instances where at most $m$ balls can be present at a time. Suppose we begin in a state that contains at most $m - k$ balls, and where bin $1$ contains $k + 1$ fewer balls than each of bins $2, 3, 4$. Then there is an oblivious sequence of $O(m)$ insertions/deletions such that, after the sequence is complete, we have the following property with probability $\Omega(1)$: some bin contains $m/4 + \Omega(\sqrt{k})$ balls.
\end{proposition}
\begin{proof}
Let $ X_0 $ denote the initial state of the game. Consider the sequence with the following three steps.
\begin{enumerate}
    \item Insert $k$ balls $x_1, x_2, \ldots, x_k$ to get to a state $X_1$. 
    \item Then insert $m - j$ balls $y_1, y_2, \ldots, y_{m - j}$, where $j$ is the number of balls in state $X_1$---this brings us to a state $X_2$ with $m$ balls in total.
    \item  Finally, delete the balls $x_1, x_2, \ldots, x_k$, and insert new balls $z_1, z_2, \ldots, z_k$ to reach a state $X_3$.
\end{enumerate}  We claim that, for at least one of the two states $X_2$ and $X_3$, we have with probability $\Omega(1)$ that some bin contains $m / 4 + \Omega(\sqrt{k})$ balls. 

During the insertions of $x_1, x_2, \ldots, x_k$, we are always in a state where bin 1 contains fewer balls than bins $2, 3, 4$. Thus, each insertion $x_i$ will go into bin 1 if and only if $1 \in \{h_1(x_i), h_2(x_i)\}$ (this is where we are exploiting that the \textsc{Greedy} algorithm is too aggressive). The number $A$  of balls $x_1, x_2, \ldots, x_k$ that are placed in bin 1 is therefore given by
\[A = |\{i \mid 1 \in \{h_1(x_i), h_2(x_i)\}\}|.\]
Let $\mu = \E[A]$. As $A$ is a binomial random variable with mean $\mu = \Theta(k)$, with probability $\Omega(1)$ we have 
\[A \ge \mu + \Omega(\sqrt{k}).\]
Now consider the number $B$ of balls $z_1, z_2, \ldots, z_k$ that are placed into bin $1$. We deterministically have that
\begin{equation}
\label{eq:Bh}
B \le |\{i \mid 1 \in \{h_1(z_i), h_2(z_i)\}\}|.
\end{equation}
Since the right side of \eqref{eq:Bh} is a binomial random variable with mean $\mu$, we have with probability $\Omega(1)$ that \[B \le \mu.\]
Moreover, since $A$ and $B$ are independent, the above bounds on $A$ and $B$ hold simultaneously with probability $\Omega(1)$. 

Finally, let us consider the number of balls in bins $2, 3, 4$ once we reach state $X_3$. Assume that state $X_2$ has maximum load $m / 4 + o(\sqrt{k})$, otherwise we are already done. Then, since $X_2$ contains $m$ balls in total, bins $2, 3, 4$ must contain a total of at least $3m/4 - o(\sqrt{k})$ balls. By step 3 of the input sequence above, it follows that, in state $X_3$, the total number of balls in bins $2, 3, 4$ is at least
\[ 3m/4 - o(\sqrt{k}) - (k - A) + (k - B).\]
Conditioning on the event above, and plugging in our bounds for $A$ and $B$, we see that (with probability $\Omega(1)$) this is at least
$3m/4 + \Omega(\sqrt{k}).$
Thus, at least one of bins $2, 3, 4$ must contain
$ m/4 + \Omega(\sqrt{k})$ balls, as desired.
\end{proof}

\paragraph{The lower bound.}
Using Proposition \ref{prop:greedygap}, the claimed lower bound follows quite easily. 
Consider the following input sequence, starting from an empty system. (1) Insert $m$ balls into the system; (2) delete each ball independently and randomly with probability $1/2$; and (3) apply the sequence in Proposition \ref{prop:greedygap} with $k = \sqrt{m}$.

As the deletions are random in step (2), the precondition for Proposition \ref{prop:greedygap} (i.e.,~the least loaded bin contain at least $k=\sqrt{m}$ fewer balls than every other bins) holds with probability $\Omega(1)$. So by Proposition \ref{prop:greedygap}, we can achieve $m / 4 + \Omega(\sqrt{k}) = m / 4 + \Omega(m^{1/4})$ balls in some bin, with probability $\Omega(1)$.

\paragraph{General $n$.} For $n$ bins, where $n$ is arbitrary, the same approach gives a lower bound of 
\begin{equation}
    m/n +\Omega(m^{1/4}/\sqrt{n^3\log n}).
    \label{eq:generaln}
\end{equation}

In particular, Proposition \ref{prop:greedygap} can be directly modified, in this setting, to achieve an overload of $\Omega(k^{1/2}/n)$: instead of using $k$ balls in each of steps $1$ and $3$, use $kn/100$ balls; then by the same argument as in the lemma, we have $A -B = \Omega(k^{1/2})$ with constant probability; this means that bin 1 is under-loaded by at least $\Omega(k^{1/2})$, and thus that some other bin is over-loaded by at elast $\Omega(k^{1/2} /n)$. 

To achieve \eqref{eq:generaln} using the modified Proposition \ref{prop:greedygap}, we just need to cause the smallest load to be $k = \Theta(\sqrt{m/(n\log n)})$ smaller than the other loads---this can again be achieved again by performing $m$ insertions and then deleting each ball independently with probability $1/2$. After the $m$ insertions, every bin will have essentially the same load ($\pm O(\log \log n)$ w.h.p. in $n$). Conditioning on the loads, the number of balls deleted from each bin is a Gaussian with standard deviation $\sigma = \Theta(\sqrt{m/n})$ (and the Gaussians are independent between bins). By standard estimates on order statistics, the difference in loads between the least loaded and the second least loaded bins is roughly the difference between the $1/n$-th and $2/n$-th percentile of the distribution, see e.g.,~\cite{Roy82},  which in expectation is $\Theta(\sigma/\sqrt{\log n})$ for the Gaussian $N(0,\sigma^2)$---hence an imbalance of $k = \Theta(m / (n \log n))$.

\subsection{The Stronger $\Omega(m^{1/2})$ Lower Bound}
\label{sec:unevengreedy}
%In the rest of the section, 
We now show how to achieve the stronger bound of $m / 4 + \Omega(\sqrt{m})$ balls in some bin. 
%This bound is significant because it is the same bound that we get from the \textsc{SingleChoice} algorithm (which simply places each ball into a random bin). This means that, in the insertion/deletion setting with $4$ bins, \textsc{Greedy} is no better than \textsc{SingleChoice}.
Given Proposition \ref{prop:greedygap}, to prove Theorem \ref{thm:greedy} it suffices to show how to achieve a gap of $k = \Omega(m)$ between bin $1$ and bins $2, 3, 4$. This is accomplished in the following proposition.

\begin{proposition}
\label{prop:unevengreedy}
Consider the \textsc{Greedy} algorithm on 4 bins, with the restriction that at most $m$ balls can be present at a time. There exists an oblivious sequence of $\poly(m)$ insertions/deletions such that, after the sequence is complete, we have the following property with probability $\Omega(1)$: Bin $1$ contains $\Omega(m)$ fewer balls than each of bins $2, 3, 4$.
\end{proposition}

%Combined, the propositions imply Theorem \ref{thm:greedy}.

The rest of the section is focused on the proof of Proposition \ref{prop:unevengreedy}.

Let $0 < \epsilon_1,\epsilon_2,\epsilon_3 < 1$ be  constants, where $\epsilon_2$ is sufficiently small as a function of $\epsilon_1$, and let $\epsilon_3$ is sufficiently small as a function of $\epsilon_2$. Sometimes we will write $\epsilon_1, \epsilon_2, \epsilon_3$ inside the $O(\cdot)$ notation, to make the dependence on them explicit, while hiding fixed constants that do not depend on $ \epsilon_1, \epsilon_2, \epsilon_3$.
%\nikhil{I think we don't need the last line anymore, and all these dependences are now explicit. Should double check.}

\subsubsection{Some basic gadgets}
We begin with a basic technical lemma establishing that \textsc{Greedy} has a tendency of eliminating imbalances over time. For brevity (and since the proof follows from standard arguments), we defer the proof of Lemma \ref{lem:gaplemma} to Appendix \ref{app:gaplemma}.

\begin{lemma}\label{lem:gaplemma}
Consider the \textsc{Greedy} algorithm on 4 bins, and fix an arbitrary initial state in which the bins have loads within $\epsilon_2 m $ of each other. If $\epsilon_1 m$ insertions are performed, then after the sequence is complete, all of the bins have loads within $O(\log m)$ of each other with high probability in $m$. Furthermore, with high probability in $m$, there is a point in time prior to the final insertion at which all of the bins have equal loads. 
\end{lemma}
%\nikhil{Do we only use Lemma 10 to argue that the bin loads are symmetric. We could also directly argue that whenever loads of $j$ and $k$ become identical, they become indistinguishable (i.e. strong Markov property). }

Using Lemma \ref{lem:gaplemma}, we now construct a simple strategy for forcing \textsc{Greedy} to add a ball to a uniformly random bin.

\begin{lemma}[Uniform ball placement gadget]
\label{lem:greedyuniform}
Consider the \textsc{Greedy} algorithm on 4 bins, and fix an arbitrary initial state in which the bins have loads within $\epsilon_2 m$ of each other. Suppose we insert balls $x_1, \ldots, x_{\epsilon_1 m }$, and then we delete balls $x_1, \ldots, x_{\epsilon_1 m  - 1}$ (all except the last insertion). With high probability in $m$, this is equivalent to placing the ball $x_{\epsilon_1 m}$ uniformly at random into one of the bins $1, 2, 3, 4$.

\end{lemma}
\begin{proof}
We have by Lemma \ref{lem:gaplemma} that, with high probability in $ m $, there is some insertion $x_i$, $i \in [\epsilon_1 m - 1]$, after which the bins have equal loads. It follows that, from the perspectives of insertions $x_{i + 1}, \ldots, x_{\epsilon_1 m}$, the four bins are symmetric. Thus the last insertion $x_{\epsilon_1 m}$ is equally likely to be placed into each of the bins, which establishes the lemma.
\end{proof}

Lemma \ref{lem:greedyuniform} allows for us to place a ball into a random bin, but we can only do this $O(m)$ times before there are too many balls ($> m$) in the system. 
But for the purposes of Proposition \ref{prop:unevengreedy}, we will need to do this $\Omega(m^2)$ times.
Our next lemma provides a mechanism for reducing the number of balls that are present while having only a small effect on the relative loads of the bins.

\begin{lemma}[Almost equal load reduction gadget]
\label{lem:greedyflush}
Consider the \textsc{Greedy} algorithm on 4 bins, and fix an arbitrary initial state in which the bins $1, 2, 3, 4$ have loads $\ell_1, \ell_2, \ell_3, \ell_4$ within $\epsilon_2 m$ of each other. We can construct an oblivous sequence of $O(\epsilon_1 m)$ insertions/deletions such that, after this sequence, the total number of balls in the system is at most $\epsilon_1 m$; and such that, with high probability in $m$, the new bin loads $\ell_i'$ for $i\in [4]$ satisfy
\begin{equation}
\label{eq:strongell}
\ell_i' = \ell_i - r + Y^{(i)},
\end{equation}
where $r \in \mathbb{N}$, $|Y^{(i)}| \le O(\log m)$, and $\E[Y^{(i)}] = 0$.   
\end{lemma}
\begin{proof}
Let us begin by describing is a sequence of $O(\epsilon_1 m)$ insertions/deletions after which (1) the total number of balls in the system is at most $\epsilon_1 m$; and (2) the new loads $\ell_i'$ of the bins satisfy (w.h.p. in $m$)
\begin{equation}
\label{eq:weakerell}
\ell_i' = r - \ell_i + Y^{(i)},
\end{equation}
where $r \in \mathbb{N}$, $|Y^{(i)}| \le O(\log m)$, and $\E[Y^{(i)}] = 0$. (Note that \eqref{eq:weakerell} is the same as \eqref{eq:strongell} but with $r$ and $\ell_i$ flipped).

The lemma  would then follow by applying the above construction twice. That is, first we obtain $\ell_1', \ell_2', \ell_3', \ell_4$ satisfying \eqref{eq:weakerell}, and then apply it again to obtain $\ell_1'', \ell_2'', \ell_3'', \ell_4''$ satisfying 
\begin{equation}
\label{eq:weakerell2}
{\ell_i}'' = r' - \ell_i' + {Y'^{(i)}},
\end{equation}
where $r' \in \mathbb{N}$, $|{Y'^{(i)}}| \le O(\log m)$, and $\E[{Y
'^{(i)}}] = 0$. Chaining together \eqref{eq:weakerell} and \eqref{eq:weakerell2}, we get relationship between $\ell_1, \ell_2, \ell_3, \ell_3$ and $\ell_1'', \ell_2'', \ell_3'', \ell_4''$ as desired by \eqref{eq:strongell}. 

Our construction for achieving \eqref{eq:weakerell} is very simple: we perform $\epsilon_1 m$ insertions $x_1, x_2, \ldots, x_{\epsilon_1 m}$, and then we delete all of the \emph{other elements} besides $x_1, x_2, \ldots, x_{\epsilon_1 m}$. Let $\ell_i$ be the load of bin $i$ before these insertions/deletions, let $q_i$ be the load of bin $i$ after the insertions are completed (but the deletions have not yet begun), and let $\ell_i'$ be the load of bin $i$ after the deletions have completed. 

By Lemma \ref{lem:greedyflush}, the quantities $q_1, q_2, q_3, q_4$ are within $O(\log m)$ of each other (w.h.p.~in $m$). Moreover, w.h.p.~in $m$, there is some point during the insertions at which all of the bins have equal loads---if we condition on this, then we have $\E[q_1] = \E[q_2] = \E[q_3] = \E[q_4]$ by symmetry. Defining $Y^{(i)} = q_i - \E[q_i]$, we have $|Y^{(i)}| \le O(\log m)$, and $\E[Y^{(i)}] = 0$.   

As $\ell_i' = q_i - \ell_i$, we have
$\ell_i' = \E[q_i] + Y^{(i)} - \ell_i$.
Setting $r = \E[q_i]$, it follows that  \eqref{eq:weakerell} holds w.h.p.~in $m$. 
\end{proof}

\subsubsection{Applying the gadgets}
We say that an application of Lemma \ref{lem:greedyuniform} or of Lemma \ref{lem:greedyflush} {\em fails} if either: the precondition of $\ell_1, \ell_2, \ell_2, \ell_4$ being within $\epsilon_2 m$ of each other fails (this is a {\em precondition failure)}; or the high-probability guarantee offered by the lemma fails (this is a {\em probabilistic failure}). 

We now describe the sequence of insertions/deletions that we use to achieve Proposition \ref{prop:unevengreedy}.
We perform $\epsilon_3 m$ phases, where phase $a \in [\epsilon_3 m]$ proceeds as follows:
\begin{itemize}
\item Apply Lemma \ref{lem:greedyuniform} $m$ times, one after another. For $b \in [m]$, use $Z_{m \cdot (a - 1) + b}$ to denote the bin that the $b$-th application of the lemma adds a ball to. If the lemma fails, then for the sake of analysis, we redefine $Z_{m \cdot (a - 1) + b}$ to be uniformly random in $[4]$. This ensures that, regardless of whether the lemma fails, the $Z_{i}$'s are independently and uniformly random in $[4]$. 
\item Apply Lemma \ref{lem:greedyflush} once to reduce the loads almost equally. Let $Y^{(1)}_a,Y^{(2)}_a,Y^{(3)}_a,Y^{(4)}_a$ denote the outcomes of $Y^{(1)},Y^{(2)},Y^{(3)},Y^{(4)}$ in that application of the lemma. If the lemma fails, then for the sake of analysis, we redefine $Y^{(1)}_a,Y^{(2)}_a,Y^{(3)}_a,Y^{(4)}_a$ to be $0$. 
\end{itemize}

To analyze the sequence of insertions/deletions, we first argue that the $Y^{(s)}_i$s have a negligible effect on the loads of the bins at any given moment.
\begin{lemma}
\label{lem:Yimpact} 
Let $s \in [4]$ and $k \in [\epsilon_3 m]$.
Then w.h.p. in $m$, it holds that  for each $k$, 
$ | \sum_{a = 1}^{k} Y^{(s)}_a | \le \tilde{O}(\sqrt{m})$,  where $\tilde{O}(\cdot)$ hides polylogarithmic factors in $m$. 
\end{lemma}
\begin{proof}
The sequence of partial sums $P_r = \sum_{a = 1}^r Y^{(s)}_a$ for $r =0,\ldots,k$ forms a martingale satisfying $|P_r - P_{r - 1}|  = O(\log m)$ deterministically for each $r \in [k]$. The lemma follows from Azuma's inequality. 
%\nikhil{Strictly speaking this is not really true as if you condition on $Y_a =O(\log m)$ the expectation may not be $0$. Formally one should define a stopping time if $Y_i>\log m$. But we can ignore this.}
\end{proof}

Next we consider the effect of the $\epsilon_3 m^2$ insertions $Z_i$ over the $\epsilon m$ phases, and show that with probability at least $1 - \epsilon_2$, there is no point in time at which the $Z_i$'s cause an imbalance of more $\epsilon_2 m / 2$. %\nikhil{Something seems fishy. As $\epsilon_2$ gets smaller you get both better probability and better error bound?}\bill{The reason we aren't cheating is b/c we are constraining $k \in [\epsilon_3 m^2]$ where $\epsilon_3 \ll \epsilon_2$. Maybe we can hint at this better.}

For $k \in [\epsilon_3 m^2]$ and $s\in [4]$,  let \[S(k,s) = |\{i \in [k] \mid Z_i = s\}|\] denote the number insertions in bin $s$ during the first $k$ applications of Lemma \ref{lem:greedyuniform}. 

\begin{lemma}
\label{lem:nofailure1}
Let $s \in [4]$ and $\epsilon_2 = (2\epsilon_3)^{1/3}$. With probability at least $1 - \epsilon_2$, it holds (simultaneously) for all $k \in [\epsilon_3 m^2]$ that
\[| S(k,s) - k/4| \le \epsilon_2 m / 2.\]
%for all .\bill{It's for all $k$ simultaneously}
\end{lemma}
\begin{proof}
As $Z_i$ is equal to $s$ independently with probability $1/4$, the sequence $S(k,s)-k/4$ for $k=0,1,\ldots,\epsilon_3m^2$ forms a martingale with increments $\{-1/4,3/4\}$ (and hence variance at most $1$). By the maximal inequality for martingales, for any $\lambda >0$,
\[ \Pr\Big[ \max_{k \in [\epsilon_3 m^2]} |S(k,s)| > \lambda \Big] \leq 2 \frac{\var[S(\epsilon_3 m^2,s)]}{\lambda^2} \leq 2\frac{\epsilon_3 m^2}{\lambda^2}.\]
Setting $\lambda =  m (2\epsilon_3/\epsilon_2)^{1/2}$ so that the right hand side above is $\epsilon_2$, and choosing $\epsilon_2^3 \le 2 \epsilon_3$ so that $\lambda \ge  \epsilon_2 m$ gives the claimed result.
%\[\Pr[| S(k,s) - k/4| > \epsilon_2 m/2 ] \leq  \Pr[| S(k,s) - k/4| > \lambda] \leq  \epsilon_3 m^2/ \lambda^2 = \epsilon_2.\] \qedhere
\end{proof}

Combining Lemmas \ref{lem:Yimpact} and \ref{lem:nofailure1}, we can bound the probability of any failures occurring during our construction.
\begin{lemma}
With probability at least $1 - \epsilon_2 - 1 / \poly(m)$, no failures (either precondition failures or probabilistic failures) occur during the construction.
\label{lem:nofailure2}
\end{lemma}
\begin{proof}
Probabilistic failures occur with probability only $1 / \poly(m)$ per application of Lemma \ref{lem:greedyuniform} or Lemma \ref{lem:greedyflush}. Across the $O(m^2)$ applications of the lemmas, the probability of a probabilistic failure ever occurring is at most $1 / \poly(m)$. 
For the rest of the proof, we condition on no probabilistic failures occurring. 

We now bound the probability of any precondition failure. Before any particular application of Lemma \ref{lem:greedyuniform} or Lemma \ref{lem:greedyflush} (during the input sequence of insertions/deletions), for bin $s \in [4]$, the amount by which its load differs from the mean can be expressed as 
\[\Big| \sum_{i = 1}^{k_1} Y^{(s)}_i + S(k_2,s) - k_2/4 \Big|\]
for some $k_1, k_2$. By Lemmas \ref{lem:Yimpact} and \ref{lem:nofailure1}, the probability that this quantity ever exceeds $\epsilon_2 m$ (and hence any precondition failure occurring) is at most $\epsilon_2 + 1/ \poly(m)$,
%Thus, the probability of any precondition failures occurring is at most $\epsilon_2 + 1 / \poly(m)$,
which completes the proof.
\end{proof}

Finally, we argue that with probability at least $\epsilon_1$, the $Z_i$'s \emph{do cause} an imbalance of $\Omega(m)$ at the end of the construction.
In particular, bin 1 contains $\Omega(m)$ fewer balls than bins $2, 3, 4$. 
\begin{lemma}
\label{lem:Zimb}
With probability at least $\epsilon_1$, we have that
\[| S(\epsilon_3 m^2, 1)\}| < \max_{s \in \{2, 3, 4\}} | S(\epsilon_3 m^2, s)| - \Omega( \sqrt{\epsilon_3} m).\]
%where, in this case, $\Omega(m)$ treats $\epsilon_3$ as a constant.

\end{lemma}
\begin{proof}
%Assume without loss of generality that $m$ is at least a sufficiently large positive constant (since otherwise the lemma is trivial). Another way to think about this lemma is as follows: 
Let $X_s$ denote the number of such balls inserted in bin $s$. 
Then $X_1$ is a binomial random variable with mean $\mu = \Theta(\epsilon_3 m^2)$. Thus, with probability at least $2\epsilon_1$, we have that, $X_1 \le \mu - 10\sqrt{\mu}$. On the other hand, if we condition on some value $\le \mu - 10\sqrt{\mu}$ for $X_1$, then the variables $X_2, X_2, X_4$ become binomial random variables with means $\mu' > \mu$. Each $X_i$ has probability at least $0.9$ of satisfying $X_i > \mu' - 5 \sqrt{\mu'} \ge \mu - 5\sqrt{\mu}$. Thus, if we condition on $X_1 \le \mu - 10\sqrt{\mu}$, then the probability at least $0.7$, we have  $X_2, X_3, X_4 > \mu - 5 \sqrt{\mu}$. Putting these together, the probability that $\max\{X_2, X_3, X_4\} - X_1 > 5\sqrt{\mu}$ is at least 
\[\Pr[X_1 \le \mu - 10 \sqrt{\mu}] \cdot \Pr[X_2, X_3, X_4 > \mu - 5 \sqrt{\mu} \mid X_1 \le \mu - 10 \sqrt{\mu}] \ge 2 \epsilon_1 \cdot 0.7 > \epsilon_1. \qedhere \]
\end{proof}
We  can now complete the proof of Proposition \ref{prop:unevengreedy}.

\begin{proof}[Proof of Proposition \ref{prop:unevengreedy}]
We prove the proposition using the construction described in this section. Note that, by design, there are never more than $m$ balls present at a time, as Lemma \ref{lem:greedyflush} brings the number of balls back down to $\epsilon_1 m$ every $O(\epsilon_1 m)$ operations.

By Lemma \ref{lem:nofailure2}, with probability at least $1 - \epsilon_2 - 1 / \poly(n)$, all of the applications of Lemma \ref{lem:greedyuniform} and Lemma \ref{lem:greedyflush} succeed. Conditioned on this, at the end of the construction, the gap of each bin $s \in [4]$ can be expressed as 
\[\sum_{i = 1}^{\epsilon_3 m} Y^{(s)}_i + S(\epsilon_3 m^2,s) - \epsilon_3 m^2 /4.\]
By Lemma \ref{lem:Yimpact}, we have $\left|\sum_{i = 1}^{\epsilon_3 m} Y^{(s)}_i\right| \leq \tilde{O}(\sqrt{m})$
 with high probability in $m$. On the other hand, by Lemma \ref{lem:Zimb},
\[|S(\epsilon_3m^2,1)| < \max_{s \in \{2, 3, 4\}} | S(\epsilon_3m^2,s)| - \Omega(\sqrt{\epsilon_3} m)\]
with probability at least $\epsilon_1$. It follows that, with probability at least $\epsilon_1 - \epsilon_2 - 1 / \poly(n)$, 
the load of bin 1 at the end of the construction is $\Omega(\sqrt{\epsilon_3} m)$ smaller than the loads of bins $2, 3, 4$.
\end{proof}

%% file: generallower.tex
\section{An Impossibility Result For The Deletions with Reinsertions}\label{sec:generallower}

In this section we prove an impossibility result for the reinsertion/deletion model, namely, that no ID-oblivious insertion strategy can guarantee sub-polynomial overload. 

\begin{theorem}
Consider the reinsertion/deletion model with $4$ bins, and with a limit of up to $m$ balls present at a time. Against any ID-oblivious insertion strategy, it is possible for an oblivious adversary to force a maximum load of $m / 4 + m^{\Omega(1)}$ at some point in the first $\poly(m)$ operations, with high probability in $m$. \label{thm:generallower}
\end{theorem}

The section splits the proof of Theorem \ref{thm:generallower} into two parts. First, in Subsection \ref{sec:marblegame}, we introduce and analyze the so-called \defn{marble-splitting game}; then, in Subsection \ref{sec:gaplower} we show how to perform a sequence of insertions/deletions that simulates an instance of the marble-splitting game and forces some bin to contain load $m/4 + m^{\Omega(1)}$ with non-negligible probability.

\subsection{The Marble-Splitting Game}\label{sec:marblegame}

In this section we present and analyze a simple game, which we call the \defn{marble-splitting game}---the game plays an important role in our lower bound for balls-and-bins games with reinsertions.

In the marble-splitting game, there are two players Alice and Bob. The player Alice has two types of moves: she can perform an \textsc{Insert} operation, which adds a new marble into the game, or she can perform a \textsc{Split($x, y$)} operation, which takes two marbles $x$ and $y$ and replaces them with new marbles $x'$ and $y'$. Alice must decide her moves at the beginning of time (so she is an oblivious adversary). 

The second player Bob gets to assign a \defn{value} $v_x$ to each marble $x$, according to the following rule: whenever Alice performs an \textsc{Insert}, Bob can assign the new marble an arbitrary real-numbered value in the range $[-1, 1]$; and whenever Alice performs a \textsc{Split($x, y$)} operation, Bob assigns $x'$ and $y'$ values $v_{x'}$ and $v_{y'}$ satisfying
\begin{align}
v_{x'} + v_{y'} = v_x + v_y \pm o(R^{-2}), \quad\text{and} \quad 
v_{x'} - v_{y'} \ge 2/R.
\label{eq:splitval}
\end{align} 
Equivalently, $v_{x'} = (v_x + v_y)/2 + \Delta \pm o(R^{-2})$ and $v_{y'} = (v_x + v_y)/2 - \Delta \pm o(R^{-2})$ for some $\Delta \ge 1 / R$. 

Alice's goal is to force some marble (she need not know which one) to have a value greater than $1$ at some point within the first $O(R^3)$ steps of the game. Her disadvantage is that she does not know the precise values of marbles. Intuitively, she would like to perform split operations on marbles $x$ and $y$ that satisfy $|v(x) - v(y)| = o(1/R)$. But she might, for example, accidentally split two marbles $x$ and $y$ whose values differ considerably---this would result in $x'$ and $y'$ having values that are \emph{closer together} than $x$ and $y$ had, which is intuitively counterproductive for Alice. We shall see that, nonetheless, Alice can deterministically force a win within $O(R^3)$ steps.

In constructing Alice's strategy, we will find it helpful for accounting purposes to artificially place the following additional constraints on Alice. We think of there as being {\em bags} $0, 1, 2, \ldots$, each of which is capable of holding arbitrarily many marbles. Whenever a marble is inserted, we place it in bag $1$. Whenever a split \textsc{Split($x, y$)} operation is performed, we require that the marbles $x$ and $y$ are currently in the \emph{same} bag $i \ge 1$ as each another, and after the split, we place the new marbles $x'$ and $y'$ into bags $i + 1$ and $i - 1$, respectively. This restriction somewhat limits Alice's possible strategies, but, as we shall see, it also simplifies the task of analyzing Alice's ``progress'' over time. 

The key result of this section is the following.
\begin{proposition}
Alice can deterministically force some marble to have a value greater than $ 1 $ at some point within the first $ O (R^3)$ steps of the game. Moreover, the strategy performs only $O(R^2)$ insertions.
\label{prop:marbles}
\end{proposition}
\begin{proof}
We begin by describing Alice's strategy. Let $ c $ be a large positive constant. She initially performs one insertion into bag $1$. She then proceeds in $c R$ phases, where at the beginning of phase $i \in \{1, 2, \ldots, cR\}$, the state of the system is as follows: bag 0 contains some arbitrary number of marbles; bags $1, 2, \ldots, i$ each contain one marble; and bags $i + 1, i + 2, \ldots$ are empty. 

The $i$-th phase consists of $(i + 1)$ sub-phases, where at the beginning of each subphase $j \in \{1, 2, \ldots, i + 1\}$, the state of the system is as follows: bag 0 contains some arbitrary number of marbles;  
%bags $1, 2, 3, 4, \ldots, i - j + 1$ each contain one marble; bag $i - j + 2$ is empty; and bags $i - j + 3, \ldots, i + 1$ each contain one marble. In other words,
and, with the exception of bag $i - j + 2$, which is empty, all of bags $2, 3, 4, \ldots, i + 1$ contain one marble (so bags $1, 2, 3, 4, \ldots, i - j + 1$ each contain one marble; bag $i - j + 2$ is empty; and bags $i - j + 3, \ldots, i + 1$ each contain one marble).

The $(i + 1)$-th subphase is special in that, all Alice does is perform one more insertion in order to reach the starting state for phase $i + 1$ (i.e., all of bags $1, 2, \ldots, i + 1$ contain 1 marble). 

For $j < i + 1$, the $j$-th subphase of phase $i$ is implemented as follows. Alice inserts one marble into bag 1. She then performs splits, one after another, on bags $1, 2, 3, \ldots, i - j + 1$. For each $t \in \{1, 2, \ldots, i - j\}$ (i.e., for every split but the final split), after she performs a split on bag $t$, the state of the system is that: bags $1, 2, \ldots, t - 1$ contain one marble each; bag $t$ is empty; bag $t + 1$ contains $2$ marbles; and bags $t + 2, t + 3, \ldots$ are as they were at the beginning of the subphase.  The final split that Alice performs (i.e., the split in $i - j + 1$) has the effect of placing a marble into the previously empty bags $i - j$ and $i - j + 2$, and leaving bag $i - j + 1$ as the solitary empty bag out of bags $0, 1, 2, \ldots, i + 1$. Thus we reach the starting state for the $(j + 1)$-th subphase. 

%\noindent{\bf Analysis of the strategy.} 
\paragraph{Analysis of the strategy.} 
The analysis will need only the following basic facts about Alice's strategy:
(1) it performs a total of $O(c^2R^2)$ \textsc{Insert} operations and $\Omega(c^3R^3)$ \textsc{Split} operations; (2) it only places marbles in bags $i \le cR + 1$; and (3)  at the end of the game, there is at most 1 marble in each bag $i$ for $i>0$.

Let $B_i$ denote the marbles in bag $i$ at any given moment, and define the potential function
\[\phi = \sum_{i = 0}^\infty i \cdot \sum_{x \in B_i} v_x.\]
We will prove the proposition by analyzing how $\phi$ evolves over time. 

Each time that an \textsc{Insert} is performed, $\phi$ may decrease by up to $1$, as $v_x \in [-1,1]$ and the marble is inserted in bag 1. During the entire game, this leads to a decrease of at most 
$O(c^2R^2)$.

Each time that a \textsc{Split} is performed, two marbles $x$ and $y$ in some bag $i$ are replaced by $x'$ and $y'$ with values given by \eqref{eq:splitval}.
%in bags $i+1$ and $i-1$, 
%with
%\[v_{x'} = (v_x + v_y)/2 + \Delta \pm o(R^{-3}) \quad \text{and}
% \quad v_{y'} = (v_x + v_y)/2 - \Delta \pm o(R^{-3}),\]
%for some $\Delta \ge 1 / R$. 
Removing $x$ and $y$ decreases $\phi$ by $i \cdot (v_x + v_y)$ and inserting $x'$ and $y'$ increases $\phi$ by
\[(i+1) v_{x'} + (i-1) v_{y'} = i \cdot (v_x + v_y) + (v_{x'}-v_{y'}) \pm o(i R^{-2}) \geq  i \cdot (v_x + v_y) + 1/R.\]
%\nikhil{I think proof also works with $R^{-2}$ instead of $R^{-3}$, which should give $m^{1/4}$ bound later instead of $m^{1/6}$, in case we care.}
% But the addition of balls $x'$ and $y'$ increases $\phi$ by
% \begin{align*}
% & (i - 1) \cdot \left((v_x + v_y)/2 - \Delta \pm o(R^{-3})\right) + (i + 1) \cdot \left((v_x + v_y)/2 + \Delta \pm o(R^{-3})\right) \\
% & = i \cdot (v_x + v_y)  \pm o(iR^{-3}) + \Delta \\
% & \ge i \cdot (v_x + v_y)  - o(1/R) + 1/R \\
% & \ge  i \cdot (v_x + v_y) + 1 / (2R).
% \end{align*}
The net effect of a split is therefore to increase $\phi$ by at least $1 /R$. As there are $\Omega(c^3 R^3)$ split operations across the entire game, this increases $\phi$ by 
$\Omega(c^3 R^2)$.

Combining the bounds for \textsc{Insert} and \textsc{Split} operations, we have that, at the end of the game,
\[\phi \ge \Omega(c^3 R^2) - O(c^2R^2) = \Omega(c^3 R^2).\]
But this means that some bin $i \le cR + 1$ must satisfy
$i \cdot  \sum_{x \in B_i} v_x > \Omega(c^2 R)$,
and thus that
$\sum_{x \in B_i} v_x > \Omega(c)$.
As $|B_i| \le 1$, this implies that there is a ball $x$ with $v_x > 1$.
\end{proof}
\noindent{\bf Remark.} It is worth noting that, in the strategy in Proposition \ref{prop:marbles}, we could have alternatively performed all of the insertions into bag 1 up front (i.e., at the beginning of the game), and then applied the appropriate \textsc{Split} operations without performing any further insertions---each marble would simply remain in bag $1$ until it was used for the \textsc{Split} operations involving it. This perspective will be convenient in our application of marble splitting. 

\subsection{Proof of Theorem \ref{thm:generallower}}\label{sec:gaplower}

We will now derive a sequence of insertions/deletions that can be used to establish Theorem \ref{thm:generallower}.

As notation, let $Q =  \{(i, j) \mid i, j \in [4], i \neq j\}$, and let $h$ be a fully independent hash function mapping each ball $x$ to a uniformly random pair $h(x) = (h_1(x), h_2(x)) \in Q$. Notice that $|Q|=12$.

insertion
\begin{align}
 \label{eq:E1}
 & |    \{x \in A \mid h(x) = (1, 2)\}| = k / 12 + t \pm O(\sqrt k) \\ 
 \text{and} \quad 
&  |\{x \in A \mid h(x) = (3, 4)\}| = k / 12 - t \pm O(\sqrt k). \label{eq:E2}
 \end{align} 
% \begin{equation}
% \label{eq:E1}
% |    \{x \in A \mid h(x) = (1, 2)\}| = k / 12 + t \pm O(\sqrt k)
% \end{equation}
% and
% \begin{equation}
% 
% |\{x \in A \mid h(x) = (3, 4)\}| = k / 12 - t \pm O(\sqrt k).
% \end{equation} 
Note that $\mathcal{E}$ only depends on the hash values for balls in $A$.

We will show that, if we condition on $\mathcal{E}$ occurring, and if $t$ is moderately large (i.e., $c\sqrt{k}$ for some sufficiently large positive constant $c$), then we can perform a sequence of insertions/deletions that make use of the sets $A$ and $B$ in order to defeat any ID-oblivious insertion strategy. 
While $\mathcal{E}$ only has a small constant probability of occurring, this can be amplified by repeating the strategy multiple times.

As a final but crucial piece of notation, for any set $S$ of balls present in the system, define the \defn{value} $v(S)$ to be the number of balls $x \in S$ that reside in bins $1, 2$. The ultimate structure of our analysis will be to show that, if an ID-oblivious algorithm guarantees a maximum load of $m/4 + m^{o(1)}$ (with high probability), then we can construct a set $S$ for which we can derive the clearly false assertion that $\E[v(S)] > |S|$.

\subsubsection{Some basic gadgets}
We will now prove a series of lemmas showing how to construct a malicious sequence of insertions/deletions using the sets $A$ and $B$ (and conditioned on $\mathcal{E}$). We begin by observing what happens if we simply insert the elements $A \cup B$ in a random order.

\begin{lemma}
\label{lem:insertAB}
Consider a balls-and-bins game with $4$ bins, starting from an arbitrary state. Suppose balls are allocated to bins using an arbitrary ID-oblivious insertion strategy that has already been shown the sets $A, B$ (i.e., the algorithm can depend on the multisets $\{h(x) \mid x \in A\}$ and $\{h(x) \mid x \in B\}$). Condition on event $\mathcal{E}$, and suppose that we insert the balls $A \cup B$  in a random order. Then, after the insertions are completed, we have 
\[\E[v(A) - v(B)] \ge t - O(\sqrt{k}).\]
\end{lemma}

The intuition behind Lemma \ref{lem:insertAB} is quite simple. For $(i, j) \in Q$, define $A_{i, j}$ (resp. $B_{i, j}$) to be the set of balls in $A$ (resp. $B$) that hash to the bin pair $(i, j)$. Due to event $\mathcal{E}$, we have that $\E[|A_{1, 2}| - |B_{1, 2}|] \ge t - O(\sqrt{k})$, so this immediately gives $A$ an extra $t - O(\sqrt{k})$ balls (in expectation) in bins $1, 2$ that $B$ doesn't get. On the other hand, for each $(i, j) \in Q \setminus \{(1, 2), (3, 4)\}$, we expect the number of balls from $A_{i, j}$ that are in bins $1, 2$ to be roughly the same as the number of balls from $B_{i, j}$ that are in bins $1, 2$, hence the conclusion of the lemma. Formalizing this argument requires some  care as the algorithm can try to distinguish the balls in $A$ from those in $B$ based on the differences between $|A_{i, j}|$ and $|B_{i, j}|$, for $(i, j) \in Q$. Thus we defer the full proof of the lemma to Appendix \ref{app:insertAB}.

Our next lemma makes a simple observation about what happens when we remove a set $X$ of balls and replace it with a set $X'$ of balls, in a balls-and-bins game that is at capacity (i.e., contains $m$ balls). 

\begin{lemma}
\label{lem:swapout}
Consider a balls-and-bins game with 4 bins, starting with $m$ balls in the system. Let $X$ be a set of $r$ balls that are present. Suppose that we delete the balls $X$, and then insert new balls $X'$, where $|X'| = r$. Then one of the following events must occur:
\begin{itemize}
\item there is some point in time at which some bin contains $m / 4 + \omega(\sqrt{k})$ balls;
\item or, $|v(X) - v(X')| = O(\sqrt{k})$. 
\end{itemize}
\end{lemma}
\begin{proof}
Suppose that they are never more than $m / 4 + \Omega(\sqrt{k})$ balls in any given bin. This means that, whenever there are $m$ balls in the system, the number of balls in bins $1, 2$ must be within $O(\sqrt{k})$ of $m/2$. 

When we remove balls $X$, we decrease the number of balls in bins $1, 2$ by $v(X)$. When we insert balls $X'$, we increase the number of balls and bins $1, 2$ by $v(X')$. In total, we must change the load of bins $1, 2$ by $O(\sqrt{k})$, meaning that $|v(X) - v(X')| = O(\sqrt{k})$.
\end{proof}

\paragraph{Gadget for splitting.}
By combining the previous two lemmas in the right way, we can construct a sequence for {\em splitting} a set $X$ of size $\poly(k)$ into two sets $Y$ and $Z$ such that $v(Y) + v(Z) = (1 \pm o(k^{-1})) v(X)$ and $\E[v(Y) - v(Z)] \ge \Omega(|X| / \sqrt{k})$. 

\begin{lemma}[Splitting gadget]
\label{lem:splittingballs}
Consider a balls-and-bins game with $4$ bins, starting from an arbitrary state with $m$ balls, and where balls are allocated to bins using an ID-oblivious insertion strategy that, as in Lemma \ref{lem:insertAB}, has already been shown the sets $A, B$, and that keeps the load of each bin below $m/4 + O(\sqrt{k})$ w.h.p. in $m$. Finally, condition on event $\mathcal{E}$ with $t = c \sqrt{k}$ for some sufficiently large constant $c>0$.

Let $X$ be a set of $q = k^{1.5} \log k$ balls that are currently present in the system. There exists a sequence of $\poly(k)$ insertions/deletions that (without ever placing more than $m$ balls in the system at a time) replaces $X$ with $q/2$-element sets $Y, Z$ satisfying
\begin{equation} \label{eq:sp1}
\E[v(Y) - v(Z)] \ge k \log k,
\end{equation}
and satisfying
\begin{equation}\label{eq:sp2}
\E[v(Y) + v(Z)] = v(X) \pm O(\sqrt{k}).
\end{equation}
\end{lemma}
\begin{proof}
Roughly speaking, the goal is to transfer the {\em imbalance} between the sets $A$ and $B$ (in how they allocate balls to bins 1,2 vs.~3,4) to the set $X$, so that the resulting sets $Y$ and $Z$ have similar {\em relative} imbalance to what $A$ and $B$ have. Of course, $A$ and $B$ have size $k$ each, while $X$ has size $q= k^{1.5}\log k$, so the imbalance between $A$ and $B$ needs to be amplified in order to get the same relative imbalance between $Y$ and $Z$. As we shall see, this is where we crucially make use of the ability to delete and \emph{reinsert} $A \cup B$ multiple times.\footnote{The other place where we make use of reinsertions is that, ultimately, we will apply Lemma \ref{lem:splittingballs} multiple times, and we will continue to reuse $A$ and $B$ across those multiple applications.}

Let us partition $ X $ into sets $X_1, X_2, \ldots, X_{q / k}$  of size $k$ each. For each $i\in [q/2k]$, we will replace  $X_{2i-1}$ by a new set $Y_i$ and $X_{2i}$ by a new set $Z_i$, in such a way that the relative imbalance between $Y_i$ and $Z_i$ is similar to that between $A$ and $B$. 
This is accomplished by performing the following sequence of insertions and deletions:
\begin{enumerate}
\item Delete the balls $X_{2i - 1} \cup X_{2i}$.
\item Insert the balls $A \cup B$ in a random order.
\item Delete the balls of $A$, and replace them with a set $Y_{i}$ of $k$ elements.
\item Delete the balls of $B$, and replace them with a set $Z_{i}$ of $k$ elements.
\end{enumerate}
By Lemma \ref{lem:insertAB}, we have after Step (2) that
\[\E[v(A) - v(B)] \ge t - O(\sqrt{k}).\]
By Lemma \ref{lem:swapout} (and since the insertion strategy keeps bin loads of $m/4 + O(\sqrt{k})$ with high probability in $m$), we then have that $\E[v(Y_i)]$ and $\E[v(Z_i)]$ are within $O(\sqrt{k})$ of $\E[v(A)]$ and $\E[v(B)]$, respectively. Thus
\[\E[v(Y_i) - v(Z_i)] \ge t - O(\sqrt{k}) \ge  2\sqrt{k},\]
where the final inequality uses the fact that $t = c \sqrt{k}$ for a sufficiently large positive constant $c$.

Summing over $i \in \{1, 2, \ldots, q / (2k)\}$, and denoting $Y = \cup_i Y_i$ and $Z = \cup_i Z_i$, we get the claimed bound
\[\E[v(Y) - v(Z)] = \sum_i \E[v(Y_i) - v(Z_i)] \ge k \log k.\]
Next, applying Lemma \ref{lem:swapout} with $X'= Y \cup Z$, we have that either
\[v(Y) + v(Z) = v(X) \pm O(\sqrt{k}),\]
or that there is some point in time at which a bin has load $m/4 + \omega(\sqrt{k})$. Since the latter event is assumed to occur with probability at most $1 / \poly(m)$, this completes the proof of the lemma.
\end{proof}

\subsubsection{Connection to marble-splitting}
We are now ready to prove Theorem \ref{thm:generallower}. We begin by proving a slightly weaker version of the theorem, namely that no ID-oblivious insertion strategy can offer a high-probability guarantee of achieving overload $m^{o(1)}$. 

\begin{proposition}
Consider the reinsertion/deletion model with $4$ bins, and with a limit of up to $m$ balls present at a time. Suppose there is an ID-oblivious bin-allocation algorithm that, for the first $\poly(m)$ steps, bounds the load of each bin by $m/4 + f(m)$ with high probability in $m$. Then $f(m) = m^{\Omega(1)}$.
\label{prop:generallower}
\end{proposition}
\begin{proof}
Set $k = m^\epsilon$ for a positive constant $ \epsilon$ to be selected later in the proof, and suppose for contradiction that $f(m) = O(\sqrt{k})$.

Let $A$ and $B$ be disjoint sets of $k$ balls each. Let $c$ be a sufficiently large positive constant, and set $t = c \sqrt{k}$. Finally, let $\mathcal{E}$ be the event that \eqref{eq:E1} and \eqref{eq:E2} hold. Note that $\mathcal{E}$ occurs with probability $\Omega(1)$; for the rest of the proof, condition on $\mathcal{E}$. 

Let $X_1, X_2, \ldots, X_{ck}$ be disjoint sets of $(k^{1.5} \log k) / 2$ balls each. To begin, insert $m$ balls into the system, where those balls include $X_1, X_2, \ldots, X_{ck}$. The sets $X_1, X_2, \ldots, X_{ck}$ will act as \emph{marbles} in a marble-splitting game. There are two types of operations that we will perform in this game: an \textsc{Insert} operation, which adds one of the sets $X_1, X_2, \ldots, X_{ck}$ as a new marble in the game; and a \textsc{Split$(X, Y)$} operation, which takes two sets $X$ and $Y$ of size $(k^{1.5} \log k) / 2$ balls each, and applies Lemma \ref{lem:splittingballs} to replace them with sets $X', Y'$ (also of $(k^{1.5} \log k)/2$ balls each) satisfying
\begin{align*}
& \frac{\E[v(X')]}{|X'|} -  \frac{\E[v(Y')]}{|Y'|} \ge 2 / \sqrt{k}, \text{ \phantom{ffffffffffffffffffffffffffffffffffffffffff} (by \eqref{eq:sp1})} \\
& \frac{\E[v(X)]}{|X|} + \frac{\E[v(Y)]}{|Y|} = \frac{\E[v(X')]}{|X'|} + \frac{\E[v(Y')]}{|Y'|} \pm o(1/k). \text{ \phantom{ffffffffffff} (by \eqref{eq:sp2})}
\end{align*}

If we define $v_X := \frac{\E[v(X)]}{|X|}$ for each set $X$ of $k^{1.5}/ 2$ balls, it follows that we are playing a marble-splitting game with $R = \sqrt{k}$, and where marbles correspond to sets of $(k^{1.5} \log k) / 2$ balls. By Proposition \ref{prop:marbles}, there is an $O(R^3) = O(k^{1.5})$-step strategy that results in some marble $X$ satisfying $v_X > 1$. This is a contradiction, since $v_X$ must deterministically be in the range $[0, 1]$.

Note that the marble-splitting game requires $O(R^2) = O(k)$ marbles at a time, each of which consists of $O(k^{1.5} \log k)$ balls. Thus, the entire game uses $O(k^{2.5} \log k)$ balls, meaning that we can set $k = m^{1/{2.5} - o(1)}$. We can therefore conclude that $f(m)$ must be at least $m^{1/5 -o(1)}$.
\end{proof}

Finally, we prove Theorem \ref{thm:generallower} by applying a basic amplification argument to Proposition \ref{prop:generallower}.  

\begin{proof}[Proof of Theorem \ref{thm:generallower}]
By Proposition \ref{prop:generallower}, there exists a parameter $s \in \poly(m)$ such that, within $\poly(m)$ operations, an oblivious adversary can achieve maximum load $m/4 + m^{\Omega(1)}$ with probability $1/s$. By independently repeating this construction $\Theta(s \log n) = \poly(m)$ times, the probability of achieving a load of $m / 4 + m^{\Omega(1)}$ at some point during the sequence becomes 
$$1 - (1 - 1 / s)^{\Theta(s \log n)} = 1 - 1 / \poly(n),$$
as desired.
\end{proof}

%% file: extensions.tex
\section{Generalizations of \textsc{ModulatedGreedy }}\label{sec:applications}

We now generalize the \textsc{ModulatedGreedy} algorithm from Section \ref{sec:upper} in  several interesting ways: 

\begin{enumerate}\setlength{\itemsep}{0pt}%
    \setlength{\parskip}{0pt}%
\item We give guarantees over an infinite time horizon, instead of $\poly(m)$ steps.
\item  We allow $m$ (the maximum number of balls present in the system) to increase with time, and only require an a-priori bound $M$ on $m$.
\item  We consider the more general $(1+\beta)$-choice and the graphical 2-choice settings 
(defined in Section \ref{sec:ext}) and extend the previous results for these settings (which were insertion-only) to also handle deletions.
\end{enumerate}

%while bounding the maximum load at any time with respect to the current baseline of $m/n$.
%+ O(\log M)$

%case to the  dynamic heavily loaded case. 

These generalizations require extending both the algorithm and the analysis techniques. We begin in Subsection \ref{sec:algoverview} by describing the algorithm and giving an overview of the key ideas; we then present the analysis and applications in Subsections \ref{sec:analysisalg} and \ref{sec:ext}.

%We describe the algorithm below, and give an overview of the key ideas.  \bill{Confusing, since it sounds like the tnire rest of the section is going to be an overview}

\subsection{The  Algorithm and Overview}\label{sec:algoverview}

%Instead of describing each of these extensions separately, we first the 
The algorithm, which we call \textsc{GeneralizedModulatedGreedy}, is described as Algorithm \ref{alg:generalized} below. Its key properties are summarized in the following theorem.

\begin{theorem}
\label{thm:generalizedupper}
Consider the insertion/deletion model with $n$ bins, and an arbitrarily long sequence of insertions/deletions, with no more than $M$ balls present at a time.
Suppose the parameters $n, M, \epsilon$ are known to the algorithm.
%If insertions are implemented using the 
Then the \textsc{GeneralizedModulatedGreedy} algorithm satisfies the 
%(Algorithm \ref{alg:generalized}), 
following guarantees:
%hold:
\begin{itemize}
\item {\em Bounded Load: } At any given moment, every bin has load at most $m/n + O(\epsilon^{-1} \log M)$ with high probability in $M$,
where $m$ is the largest number of balls that were ever present so far.
\item {\em Bounded Bias: } For any given insertion, if $i$ and $j$ are the two bins being chosen between, then each bin is selected with a probability in the range $[1/2 - \epsilon, 1/2 + \epsilon]$. 
\end{itemize}
\end{theorem}

\begin{algorithm}[h!]
\begin{algorithmic}
\Procedure{GeneralizedModulatedGreedy}{} 
\State For $k \in [n]$, let $\ell_k$ denote \# balls with color $k$. Let $\overline{\ell} = \frac{1}{n}\sum_k \ell_k$.
\State Let $m$ be the largest number of balls that have been present in the system at once thus far.
\State Let $\Delta = c \epsilon^{-2} \log M$. 
\State Set $T = \lceil m/n \rceil + \Delta - \overline{\ell}$. 
\State Select two bins $i, j \in [n]$ independently and uniformly at random. 
\If{$\left(\max_k \ell_k\right) - \left( \min_k \ell_k \right) \le \epsilon T$}
 % \State{Flip coin with probability $1/2 + \frac{\ell_j - \ell_i}{2T}$ of landing heads.}
 % \If{Heads}
        \State{With probability $1/2 + \frac{\ell_j - \ell_i}{2T}$, assign the ball to bin $i$ and assign it color $i$.}
%    \Else
        \State{Otherwise, assign the ball to bin $j$ and assign it color $j$.}
%\EndIf
\Else
   \State{Declare the ball to be \defn{corrupted}}.
   \State{Select $\rho \in [n]$ such that, for each $k \in [n]$, \[\Pr[\rho = k] = \frac{\lceil m/n \rceil + \Delta - \ell_k}{n \cdot T}.\]}
   \State{Assign the ball uniformly at random in $\{i, j\}$ and assign it color $\rho$.}
\EndIf
\EndProcedure
\end{algorithmic}
\caption{The \textsc{GeneralizedModulatedGreedy} algorithm. 
%for implementing \textsc{Insert}. 
The algorithm has parameters $M$ (an upperbound on the number of balls that will ever be present) and $\epsilon$, and makes use of a sufficiently large constant $c>0$. The algorithm outputs a bin and a color for the ball being inserted.}
\label{alg:generalized}
\end{algorithm}

Notice that the algorithm assigns a ball both a bin and a color. Typically, the color is the same as the bin to which the ball is assigned, but occasionally a ball will get \emph{corrupted}, in which case the bin and color may differ.Moreover, at any time, the maximum load is bounded with respect to $m/n$ (instead of $M/n$).

%\paragraph{Overview of techniques.}
Before giving the detailed analysis, we briefly describe the new ideas we need over those in Section \ref{sec:upper}.

\paragraph{Infinite time horizon.}
A key feature of the algorithm is that it offers guarantees on an infinite time horizon. To achieve this we explicitly incorporate the coupling with the stone game into the design of the algorithm. In particular, whenever there is an insertion that \textsc{ModulatedGreedy} would have been at risk of halting on, \textsc{GeneralizedModulatedGreedy} instead declares that ball to be \defn{corrupted}. The algorithm then ``fudges'' its bookkeeping: it treats the corrupted ball as being placed into whichever bin is necessary to maintain the coupling with the stone game.

More concretely, we assign each ball both to a bin (where it truly resides) and to a color (which, if the ball is corrupted, may differ from the ball's bin). The algorithm makes all of its decisions based on ball colors (and ignores the actual bins that balls reside in). This allows for the algorithm to maintain a coupling forever between the colors of its balls and the colors of the balls in the stone game.

\paragraph{Increasing $m$.}
Another interesting feature is that the algorithm allows for $m$ to grow over time, subject only to the constraint $m \le M$. To handle this, \textsc{GeneralizedModulatedGreedy} bases its allocation decisions on the largest value of $m$ that it has witnessed so far.
At first glance, this seems to significantly break the relationship between the balls-and-bins game and the stone game, and indeed Lemma \ref{lem:randomactive} no longer holds---however, as we shall see, the stone game and its analysis can be modified to also handle the incremental growth in $m$ over time.

\paragraph{Bias, $(1+\beta)$-choice and graphical process.}
Finally, a third feature of the algorithm is that it introduces a new variable $\epsilon$ that constrains the amount of bias that the algorithm is permitted to exhibit. We will see at the end of the section that this seemingly minor modification allows us to extend the algorithm to the $(1+\beta)$-choice and the graphical $2$-choice process, both of which are generalizations of the classical 2-choice process.
Moreover, the guarantees of the resulting algorithms matches the previous known results for the insertion-only case for these settings.

%ends up allowing for a number of interesting applications. 

\subsection{Algorithm Analysis}\label{sec:analysisalg}
We now turn to proving Theorem \ref{thm:generalizedupper}. We begin by defining the generalized stone game, which extends the stone game in Section \ref{sec:upper}.
Then we show how this game is closely related to the balls and bins game and use this relationship to analyze \textsc{GeneralizedModulatedGreedy}.

\subsubsection{The generalized stone game}
The \textsc{$\Delta$-generalized stone game} has an inactive bag and an active bag.
The inactive bag is initialized to contain $\Delta \cdot n$ stones $x_{k, j}$ for $k \in [n]$ and  $q \in [\Delta]$, and the active bag is initialized to be empty. 
We say that the ball $x_{k, q}$ has {\em color} $k \in [n]$.
The game supports two operations that are performed by an oblivious adversary: \textsc{Activate()} and \textsc{Deactivate$(r)$}. 

The \textsc{Activate()} operation (described formally in Algorithm \ref{alg:activate}) takes two steps: First, the operation moves a random stone from the inactive bag to the active bag. Second, if there are fewer than $\Delta \cdot n$ stones in the inactive bag, then it computes the number $Q \cdot n$ of stones currently in the system (active and inactive bags), and it adds $n$ new stones $\{x_{k, Q + 1}\}_{k \in [n]}$, one of each color, to the inactive bag.
This second step is different from the standard stone game in Section \ref{sec:upper}, and in particular, the total number of stones now can increase over time (in increments of $n$).

The \textsc{Deactivate($r$)} operation works exactly as before---it takes whichever stone was added to the active bag $r$-th most recently, and moves that stone back to the inactive bag. 
%\nikhil{This $r$ is too close to the $r$ in the para above. In Section 2 we used $Q$. Change later to be consistent.}

\begin{algorithm}
\begin{algorithmic}
\Procedure{Activate}{} 
\State Move a random stone from the inactive bag to the active bag.
\If{Inactive bag contains fewer than $\Delta \cdot n$ balls} 
   \State{Let $Q \cdot n$ be \# stones currently in the system}
    \State Add a \defn{batch} $B_{Q + 1} = \{x_{k, Q + 1}\}_{k \in [n]}$ of $n$ new balls to the inactive bag. 
\EndIf
\EndProcedure
\end{algorithmic}
\caption{The \textsc{Activate} method for the generalized stone game. The algorithm has parameter $\Delta$.  The moves a random stone from the inactive bag to the active bag, and then (possibly) adds additional stones to the inactive bag. }
\label{alg:activate}
\end{algorithm}

 %To analyze \textsc{GeneralizedModulatedGreedy}, 
 We begin by proving a basic fact about the generalized stone game.

\begin{lemma}
\label{lem:generalizedstonebound}
Let $c>0$ be a sufficiently large constant, and let $\epsilon, M$ be parameters. Fix any time in the $(c\epsilon^{-2}\log M)$-generalized stone game, and for $k\in [n]$, let $s_k$ denote the number of stones with color $k$ in the inactive bag. With probability $M^{-\Omega(c)}$, for each $k\in [n]$, we have that 
\[(1 - \epsilon/2) \E[s_k] \le s_k \le (1 + \epsilon/2) \E[s_k].\]
\end{lemma}
\begin{proof}
Let $Q \cdot n$ be the number of stones currently in the system. For each $q \in \{1, 2, \ldots, Q\}$, define  $B_q = \{x_{k, q}\}_{k \in [n]}$. The $n$ stones in $B_q$ are all inserted into the system in the same instant  and are indistinguishable from one another 
in terms of how they interact with the sequence of operations being performed. If there are $a_k$ balls from $B_k$ in the inactive set, then the probability that any of them have color $i$ is simply $a_k / n$.

Thus, if we fix some outcome for the values of the $a_k$'s, then we can write
$s_k =\sum_{q = 1}^Q A_{k}$,
where $A_k$ are independent indicator random variables with $\Pr[A_{q} = 1] = a_q / n$. Using $I$ to denote the set of balls in the inactive set,  the expected value of $s_k$ evaluates to
\[\E[s_k]=\sum_{q = 1}^Q a_q/n = |I|/n.\]

By design, however, the inactive set always at least $|I| \geq \Delta \cdot n = c \epsilon^{-2} n\log M$ balls, so that $\E[s_k] \ge  \Omega(c \epsilon^{-2} \log M)$. Applying a Chernoff bound (and as $c$ is a large constant), for each $k\in [n]$, $s_k$ lies between $(1 - \epsilon/2) \E[s_k]$ and $(1 + \epsilon/2) \E[s_k]$ with probability $M^{-\Omega(c)}$. 
\end{proof}

\subsubsection{Coupling with \textsc{GeneralizedModulatedGreedy}}
Next we establish the connection between the generalized stone game and  the \textsc{GeneralizedModulatedGreedy} algorithm. 

First, as in Section \ref{sec:upper}, the oblivious sequences of insertion/deletions for the balls-and-bins game maps to an input sequence of the  $\Delta$-generalized stone game as follows:
each insertion in the balls-and-bins game causes an activation in the stone game, and each deletion \textsc{Delete($x$)} in the balls-and-bins game causes a deactivation \textsc{Deactivate($r$)}, where $r - 1$ is the number of balls present in the balls-and-bins game that were inserted after $x$.

The following key lemma shows that the random choices in the two games can be coupled.

\begin{lemma}[Coupling]
Consider a sequence $\mathcal{S}$ of insertions/deletions in a balls-and-bins game on $n$ bins, with no more than $M$ balls present at a time. Let $G_1$ be a balls-and-bins game with operation-sequence $\mathcal{S}$, let $\Delta = c \epsilon^{-2} \log M$,
and let $G_2$ be $\Delta$-generalized stone game with operation sequence $\phi(\mathcal{S})$. 

If $G_1$ is implemented using the \textsc{GeneralizedModulatedGreedy} algorithm with parameters $M,c$ and $\epsilon$, then there exists a coupling between $G_1$ and $G_2$ such that: (1)  the number of balls with a given color $k \in [n]$ in $G_1$ always equals the number of active-bag stones with color $k$ in $G_2$; and (2) the total number $n \cdot Q$ of stones in $G_2$ always satisfies $Q = \lceil m / n \rceil + \Delta$, where $m$ is the largest number of balls ever present at once so far in the balls-and-bins game. 
\label{lem:coupling2}
\end{lemma}

\begin{proof}
Let $\ell_k$ denote the number of balls with color $k$ at any given moment and let $\overline{\ell} = \sum_k \ell_k / n$.
By Lemma \ref{lem:modulateddist} (modified so that $T = \lceil m/n \rceil + \Delta - \overline{\ell}$ and $T_k = \lceil \frac{m}{n} \rceil + \Delta - \ell_k$), we know that, on any given insertion in which \textsc{GeneralizedModulatedGreedy} does not create a corrupted ball, each color $k$ is selected with probability
\begin{equation}
\frac{T_k}{n \cdot T} = \frac{\lceil\frac{m}{n}\rceil + \Delta - \ell_k}{n \cdot T}.
\label{eq:modpr2}
\end{equation}
On the other hand, on insertions that do create corrupted balls, we have by design that \eqref{eq:modpr2} is still the probability of color $ k $ being selected. Thus, \eqref{eq:modpr2} is always the probability of any given color $ k $ being selected on any given insertion.

Next we turn our attention to the generalized stone game. By design, the number $n \cdot Q$ of stones in the generalized stone game at any given moment satisfies $Q = \lceil m / n \rceil + \Delta$, where $m$ is the largest number of balls that have ever been present at once in the balls-and-bins game. Suppose that, for each color $k$ there are $\ell_k$ stones with color $ k $ in the active set of the stone game. Then on any given activation, the probability of a ball with color $k$ being moved into the active set is 
\begin{equation}
\frac{Q - \ell_k}{n \cdot Q - \sum_i \ell_i} = \frac{\lceil\frac{m}{n}\rceil + \Delta - \ell_k}{n \cdot (\lceil\frac{m}{n}\rceil + \Delta - \overline{\ell})} = \frac{\lceil\frac{m}{n}\rceil + \Delta - \ell_k}{n \cdot T}.
\label{eq:stonepr2}
\end{equation}

The two probabilities \eqref{eq:modpr2} and \eqref{eq:stonepr2} are precisely equal. Thus, we can couple the games so that the color selected by the insertion in the balls-and-bins game is the same as the stone color selected by the activation in the stone game. 

If we implement the insertions/activations in this way, then the deletions/deactivations also become coupled: whenever a ball is deleted with a color $k$, a stone with color $k$ is removed from the active bag (in particular, the ball and stone were assigned to have the same color when they were inserted/activated previously). Thus the proof of the lemma is complete.
\end{proof}

Combining Lemmas \ref{lem:coupling2} and \ref{lem:gaplemma2}, we can bound the probability that a given ball is corrupted.

\begin{lemma}[Corruption probability]
Consider a sequence of insertions/deletions in a balls-and-bins game on $n$ bins with no more than $M$ balls ever present at a time, and suppose that insertions are implemented using the \textsc{GeneralizedModulatedGreedy} algorithm with parameters $M$ and $\epsilon$. For any given insertion, the probability that the ball being inserted is corrupted is at most $1 / \poly(M)$.
\label{lem:corrupted}
\end{lemma}
\begin{proof}
For $k \in [n]$, let $\ell_k$ denote the number of balls with color $k$. Let $\overline{\ell} = \sum_k \ell_k/n$ and let $\Delta = c \epsilon^{-2} \log M$, where $c$ is the constant used by \textsc{GeneralizedModulatedGreedy}. In order for the inserted ball to be corrupted, we would need
\begin{equation}
\label{eq:corrupted}
\left(\max_k \ell_k\right) - \left( \min_k \ell_k \right) > \epsilon T = \epsilon (\lceil m/n \rceil + \Delta - \overline{\ell}).
\end{equation}
If we couple the process to a $\Delta$-generalized stone game as in Lemma \ref{lem:coupling2}, then we have (1) that the number of balls with each color $k$ in the active bag of the generalized stone game is $\ell_k$; and (2) that the total number of stones in the generalized stone game is $n (\lceil m/n \rceil + \Delta)$. It follows by Lemma \ref{lem:gaplemma2} that, w.h.p. in $M$, 
\[(1 - \epsilon/2) \E[s_k] \le s_k \le (1 + \epsilon/2) \E[s_k],\]
where $s_k = \lceil m/n \rceil + \Delta - \ell_k$ and $\E[s_k] = \lceil m/n \rceil + \Delta - \overline{\ell}$. That is, each  $s_k$ deviates by at most $\frac{1}{2} \epsilon (\lceil m/n \rceil + \Delta - \overline{\ell})$ from its mean. The same holds for each $\ell_k$ (as $\ell_k+s_k$ is fixed), which implies that \eqref{eq:corrupted} does not occur.
\end{proof}

Finally, we can prove Theorem \ref{thm:generalizedupper}.
\begin{proof}[Proof of Theorem \ref{thm:generalizedupper}]
It suffices to prove the Bounded Load guarantee, since the Bounded Bias guarantee is hardcoded into the \textsc{GeneralizedModulatedGreedy} algorithm by design.
In particular, given the bin choices $i,j$, if the ball is not corrupted then $|\ell_i-\ell_j| \leq \epsilon T$ and it is assigned to bin $i$ with probability $1/2 + (\ell_j-\ell_i)/2T \leq 1/2 +\epsilon/2$. On the other hand if it is corrupted, then it is assigned uniformly.

Let $\Delta = c \epsilon^{-2} \log M$. Couple the balls-and-bins game to the $\Delta$-generalized stone game as in Lemma \ref{lem:coupling2}, and consider the state of both systems at some fixed point in time. 

By Lemma \ref{lem:corrupted}, we have with high probability in $M$ that there are no corrupted balls in the balls-and-bins game. Thus the number of balls in any given bin $k$ (in the balls-and-bins game) is equal to the number of active-bag stones with color $ k $  (in the generalized stone game). Moreover, if $m$ is the most balls that were ever present in the balls-and-bins game, the number of stones in the generalized stone game is $\lceil m / n \rceil + \Delta$.

Using $\ell_k$ to be the number of active-bag stones with color $k$, and $s_k$ to be the number of inactive-bag stones with color $k$, by Lemma \ref{lem:gaplemma2} we have that $s_k > (1 - \epsilon) \E[s_k] \ge (1 - \epsilon)  \Delta$,
which gives the desired bound 
\[\ell_k = \lceil m / n \rceil + \Delta - s_k \le \lceil m / n \rceil + \epsilon \Delta = m/n + O(\epsilon^{-1} \log M).\qedhere\]
\end{proof}

\subsection{Extensions}
\label{sec:ext}
We conclude the section with applications of \textsc{GeneralizedModulatedGreedy} to several more general settings.

\paragraph{$(1+\beta)$-choice process.}
The $(1 + \beta)$-choice setting was proposed by Peres, Talwar, and Wieder \cite{peres2010} as a useful generalization of the 2-choice process, where each insertion selects a random bin with probability $(1 - \beta)$, and gets to choose between two random bins $i, j$ with probability $\beta$. For any fixed $\beta <1$, they showed that in the insertion-only case, the \textsc{Greedy} algorithm achieves maximum load $m/n + \Theta(\beta^{-1} \log n)$ with high probability in $n$; this load becomes $m / n + \Theta(\beta^{-1} \log m)$  
if one wishes for a high-probability guarantee in $m$. They further proved that these bounds are optimal for any $(1 + \beta)$-choice insertion strategy.

We can directly use \textsc{GeneralizedModulatedGreedy} to construct an optimal $(1 + \beta)$-choice insertion strategy for the insertion/deletion model. 
\begin{theorem}
\label{thm:beta}
Consider a balls-and-bins game with $n$ bins and with no more than $m$ balls present at a time. In the insertion/deletion model, there exists a $(1 + \beta)$-choice algorithm that at any given moment, with probability in $m$, has maximum load
\[m/n + O(\beta^{-1} \log m).\]
\end{theorem}
\begin{proof}
If we set $\epsilon = \beta / 2$, then \textsc{GeneralizedModulatedGreedy} selects between bins $i, j$ with a probabilities in the range $1/2 \pm \epsilon$; this is equivalent to selecting a random bin (i.e., a random one of $i, j$) with probability $1 - 2\epsilon = 1 - \beta$, and then selecting between bins $i, j$ with a probabilities in the range $[0, 1]$. 

%So the coupling based proof extends to this setting in an identical manner as before. \bill{I commented this sentence because I'm pretty sure the coupling proof is not relevant---the GeneralizedModulatedGreedy algorithm is just being reinterpreted as a $1 + \beta$ choice strategy}.
\end{proof}

\paragraph{Graphical-Allocation.} Graphical allocation is another generalization of the $2$-choice model, introduced by Kenthapadi and Panigrahy \cite{KP06}. 
 Here we are given an arbitrary fixed $d$-regular graph $G$ on $n$ vertices (i.e., bins). To assign a ball to a bin, we select a uniformly random edge $e = (v_1, v_2)$ choose one of bins $v_1, v_2$. The classic $2$-choice process corresponds to the complete graph $G=K_n$.
 
 Bansal and Feldheim \cite{bansal2021well} showed that, in the insertion-only case, it is possible to guarantee a maximum load of $m / n + O((d / k) \log^4 n \log \log n)$ w.h.p. in $n$, where $k$ is the edge-connectivity of $G$. The linear dependence on $(d/k)$ is necessary and the bound becomes $m/n O((d / k) \log m\log^3 n \log \log n)$ if one requires the bound to be w.h.p. in $m$.

Their algorithm reduces the problem, in a black-box manner, to that of constructing a $(1 + \beta)$-choice strategy on two bins (in particular, where the two ``bins'' represent sibling sets in a binary hierarchical decomposition of the vertices of $G$, and the different sibling pairs use different choices for $\beta$, see \cite{bansal2021well}). In the insertion-only case \cite{bansal2021well}, they use the \textsc{Greedy} $(1 + \beta)$-choice strategy---to extend this to  handle deletions, we can simply use \textsc{GeneralizedModulatedGreedy} instead (as in Theorem \ref{thm:beta}). Together with the framework developed in \cite{bansal2021well}, this gives the following result.
\begin{theorem}
Consider a graphical process where, given a $k$-edge-connected $d$-regular graph $G$ on $n$ vertices (i.e., bins), the two bin choices for each ball ball are given by the endpoints of a uniformly random  edge $e = (v_1, v_2)$ of $G$.
Consider any sequence of insertions/deletions where the number of balls in the system never exceeds $m$. Then it is possible to guarantee a maximum load of $m / n + O((d/k) \log m \log^3 n \log \log n)$ w.h.p. in $m$, at any given moment. %\todo{I swapped the ns to ms here, is that right?}
\end{theorem}

%% file: appendix.tex
\section{Proof of Lemma \ref{lem:gaplemma}}\label{app:gaplemma}
We prove Lemma \ref{lem:gaplemma}, reformulated here to use a constant $c$ in place of constants $\epsilon_1, \epsilon_2$, and to use a variable $k$ in place of $\epsilon_2 m$:
\begin{lemma}[Lemma \ref{lem:gaplemma} reformulated]
Let $c>0$ be a sufficiently large constant. Consider the \textsc{Greedy} algorithm on 4 bins, and fix an arbitrary initial state in which the bins have loads within $k$ of each other. If $ck$ insertions are performed, then after the sequence is complete, all of the bins have loads within $O(\log k)$ of each other with high probability in $k$. Furthermore, with high probability in $k$, there is some intermediate point in time during which all of the bins have equal loads. 
\label{lem:gaplemma2}
\end{lemma}
We break the proof of this lemma into a few simple claims.
%We first show that after the insertions are completed, the loads are within $O(\log k)$ of each other.
\begin{claim}
\label{clm:logload}
%Let $c>0$ be a sufficiently large constant. Consider the \textsc{Greedy} algorithm on 4 bins, and 
Given an arbitrary initial state with bin loads within $k$ of each other, if $j \ge c k$ insertions are performed, then at end of the sequence, the bin loads will be within $O(\log k)$ of each other, w.h.p. in $k$.
\end{claim}
\begin{proof}
Let $D_{i, j}$ be the difference between the loads of the $i$-th and $j$-th bins (where $i \neq j$). It suffices to show that, after the insertions are complete, $D_{i, j} \le O(\log k)$ with high probability in $k$. 

Notice that whenever $D_{i, j} \neq 0$ and we insert a ball, $D_{i, j}$ has a random increment with $\Omega(1)$ bias towards $0$ (it surely decreases by $1$ when $i,j$ are the two choices, which has $\Omega(1)$ probability as $n=4$, and has zero bias otherwise). 
%Thus, $D_{i, j}$ is performing a random walk with always $\Omega(1)$ bias towards $0$, and
So starting at $|D_{i, j}| \le  k$,  w.h.p. in $k$ that the random walk thus reaches $0$ within $O(k) \le c k$ steps. Moreover, each time that the random walk hits $0$, w.h.p. in $k$ it will hit $0$ again within $O(\log k)$ steps.
%So for the $c k$ steps after the random walk hits $0$, we have with high probability in $k$ that the random walk is within $O(\log k)$ of $0$ at all times. 
Thus, after the $c k$ insertions are performed, we have $|D_{i, j}| = O(\log k)$ w.h.p. in $k$. 
\end{proof}

Next we show that, during the insertions, 
the loads become equal at some point with probability $\Omega(1)$.
\begin{claim}
\label{clm:equalloads}
%Let $c>0$ be a sufficiently large constant. Consider the \textsc{Greedy} algorithm on 4 bins, and fix 
Given any arbitrary initial state the bin loads within $k$ of each other, if $2 c k$ insertions are performed, then with probability at least $\Omega(1)$ there is some time at which all the $4$  bins have equal loads.
\end{claim}
\begin{proof}
This follows by iterated applications of Claim \ref{clm:logload}. After $ck$ insertions, all the $4$ the bins have loads within $T_1 = O(\log k)$ of each other, w.h.p. in $k$. After $c T_1$ further insertions, the bins have loads within $T_2 = O(\log T_1)$ of each other, w.h.p. in $T_1$. After $c T_2$ further insertions, the bins have loads within $T_3 = O(\log T_2)$ of each other,  w.h.p. in $T_2$. Continuing like this, after $c(k + T_1 + T_2 + \cdots + T_{O(\log^* n)}) = (c + o(1))k$ insertions, we reach a state where all bin loads are within $O(1)$ of each other with probability $\Omega(1)$. Once this occurs, we have with probability $\Omega(1)$ that during the next $O(1)$ insertions after that, there is a point at which the $4$ bins have equal loads. 
%This completes the proof of the claim.
\end{proof}

Finally, we amplify Claim \ref{clm:equalloads} in order to achieve a high-probability bound.
\begin{claim}
%Let $c$ be a sufficiently large positive constant. Consider the \textsc{Greedy} algorithm on 4 bins, and 
Given an arbitrary initial state with bin loads within $k$ of each other, if $c k$ insertions are performed, then w.hp. in $k$ there is some time when all the bins have equal loads.
\label{clm:equalloads2}
\end{claim}
\begin{proof}
By Claim \ref{clm:logload}, w.h.p. in $k$) the loads are within $T = O(\log k)$ of each other during each of the final $ck / 2$ insertions. Break these insertions into $\Omega(k / \log k)$ chunks of size $2cT$. Within each chunk, we have by Claim \ref{clm:equalloads} that the loads equalize (at some point) with probability at least $\Omega(1)$. Thus, the probability that the loads stay unequal during all  $\Omega(k / \log k)$ chunks is $\exp(-\Omega(k/\log k))$. 
%\[(1 - \Omega(1))^{\Omega(k / \log k)} < 1 / \poly(k),\]
%which completes the proof of the claim.
\end{proof}

Combined, Claims \ref{clm:logload} and \ref{clm:equalloads2} imply Lemma \ref{lem:gaplemma2}.

\section{Proof of Lemma \ref{lem:insertAB}}\label{app:insertAB}
For $(i, j) \in Q$, define $A_{i, j}$ (resp. $B_{i, j}$) to be the set of balls in $A$ (resp. $B$) that hash to the bin pair $(i, j)$. Let $a_{i, j} = |A_{i, j}|$ and $b_{i, j} = |B_{i, j}|$. Let
\[p_{i, j} = \frac{v\left(A_{i, j} \cup B_{i, j}\right)}{|A_{i, j} \cup B_{i, j}|}\]
denote the (random) fraction of balls in $A_{i, j} \cup B_{i, j}$ that are placed into bins $1, 2$. 

We remark that there are two sources of randomness in this lemma: the first, which we denote by $\mathcal{R}_1$, is the outcome of the hashes of the balls in $A$ and $B$ (i.e., the random bits that determine $\{a_{i, j}\}$ and $\{b_{i, j}\}$); the second, which we denote by $\mathcal{R}_2$, is the random order in which the balls $A \cup B$ are inserted into the system. 

Note that, from the perspective of the ID-oblivious insertion strategy, the balls $A_{i, j}$ are indistinguishable from the balls $B_{i, j}$ (this is due to the randomness from $\mathcal{R}_2$). Thus we have that, for any fixed outcome of $\mathcal{R}_1$, 
\[\E\left[v(A_{i, j}) - v(B_{i, j}) \mid \mathcal{R}_1\right] = \E[p_{i, j} (a_{i, j} - b_{i, j}) \mid \mathcal{R}_1].\]
Summing over $(i, j) \in Q$, we have that (again for any fixed outcome of $\mathcal{R}_1$)
\[\E[v(A) - v(B) \mid \mathcal{R}_1] = \sum_{(i, j) \in Q} \E\left[v(A_{i, j}) - v(B_{i, j}) \mid \mathcal{R}_1\right] = \sum_{(i, j) \in Q} \E[p_{i, j}(a_{i, j} - b_{i, j}) \mid \mathcal{R}_1].\]
Considering all outcomes for $\mathcal{R}_1$ that satisfy $\mathcal{E}$, it follows that
\[\E[v(A) - v(B) \mid \mathcal{E}] = \sum_{(i, j) \in Q} \E[p_{i, j}(a_{i, j} - b_{i, j}) \mid \mathcal{E}].\]
Thus, to prove the lemma, it suffices to show that
\[\E\left[\sum_{(i, j) \in Q} p_{i, j}(a_{i, j} - b_{i, j}) \mid \mathcal{E}\right] \ge t - O(\sqrt{k}).\]
Note that $p_{(1, 2)} =1$ and $p_{(3, 4)} = 0$ deterministically. Moreover, 
\[\E[a_{1, 2} - b_{1, 2} \mid \mathcal{E}] \ge \E[k/12 + t - O(\sqrt{k}) - b_{1, 2}] = t - O(\sqrt{k}) - \E[b_{1, 2} - k/12] = t - O(\sqrt{k}).\]
Thus  
\begin{align*}
 \E\left[\sum_{(i, j) \in Q} p_{i, j}(a_{i, j} - b_{i, j}) \mid \mathcal{E}\right] 
& =\E[a_{1, 2} - b_{1, 2} \mid \mathcal{E}] + \E\left[\sum_{(i, j) \in Q \setminus \{(1, 2), (3, 4)\}} p_{i, j} (a_{i, j} - b_{i, j})  \mid \mathcal{E}\right] \\
& = t - O(\sqrt{k}) + \E\left[\sum_{(i, j) \in Q \setminus \{(1, 2), (3, 4)\}} p_{i, j} (a_{i, j} - b_{i, j}) \mid \mathcal{E}\right] \\
& \ge t - O(\sqrt{k}) - \sum_{(i, j) \in Q \setminus \{(1, 2), (3, 4)\} } \E[|a_{i, j} - b_{i, j}| \mid \mathcal{E}].
\end{align*}
To complete the proof, it suffices to show that for each $(i, j) \in Q \setminus \{(1, 2), (3, 4)\}$, we have
\[\E[|a_{i, j} - b_{i, j}| \mid \mathcal{E}] \le O(\sqrt{k}).\]
Let $\alpha_{i, j} = \E[a_{i, j} \mid \mathcal{E}]$ and $\beta_{i, j} = \E[b_{i, j} \mid \mathcal{E}]$. By Chernoff bounds, we know that 
$\E[|a_{i, j} - \alpha_{i, j}| \mid \mathcal{E}]\le O(\sqrt{k})$ and $\E[|b_{i, j} - \beta_{i, j}| \mid \mathcal{E}] \le O(\sqrt{k})$.
Thus, it suffices to show that
\[|\alpha_{i, j} - \beta_{i, j}| = O(\sqrt{k}).\]
For each ball $x \in A$ with $h(x) \notin \{(1,2),(3,4)\}$, we have that $h(x)$ is random among the $|Q| - 2 = 10$ pairs in $Q \setminus \{(1,2),(3,4)\}$; and for each ball $x \in B$, we have that $h(x)$ is random among the $|Q| = 12$ pairs in $Q$. 
Thus $\alpha_{i, j} = \E[\frac{1}{10} (k - a_{1, 2} - a_{3, 4}) \mid E]$ and $\beta_{i, j} = k / 12$. Finally, as $a_{1, 2} + a_{3, 4} = k / 6 \pm O(\sqrt{k})$ (conditioned on event $\mathcal{E}$ occurring), we get
\begin{align*}
\alpha_{i, j} - \beta_{i, j} & = \E\left[\frac{1}{10} (k - a_{1, 2} - a_{3, 4}) \mid \mathcal{E}\right] - k / 12  = \frac{1}{10} (k - k / 6) - k / 12 \pm O(\sqrt{k})  = \pm O(\sqrt{k}),
\end{align*}
which completes the proof.